\documentclass[aps,pre,amsmath,amssymb]{revtex4-2}

\usepackage{polski}
\usepackage[english]{babel}
\usepackage{bm}
\usepackage{graphicx}
\usepackage{amsmath}
\usepackage{mathrsfs}
\usepackage{amsthm}

\newtheorem{lemma}{Lemma}

\begin{document}

\newcommand{\uu}[1]{\underline{#1}}
\newcommand{\pp}[1]{\phantom{#1}}
\newcommand{\be}{\begin{eqnarray}}
\newcommand{\ee}{\end{eqnarray}}
\newcommand{\ve}{\varepsilon}
\newcommand{\vp}{\varphi}
\newcommand{\vs}{\varsigma}
\newcommand{\Tr}{{\,\rm Tr\,}}
\newcommand{\Trr}{{\,\rm Tr}}
\newcommand{\pol}{\frac{1}{2}}
\newcommand{\sgn}{{\rm sgn}}
\newcommand{\Mo}{\mho}
\newcommand{\Om}{\Omega}

\title{On the Relativity of Quantumness as Implied by Relativity of Arithmetic and Probability}
\author{Marek Czachor}
\affiliation{
Instytut Fizyki i Informatyki Stosowanej,
Politechnika Gdańska, 80-233 Gdańsk, Poland
}

\begin{abstract}A hierarchical structure of isomorphic arithmetics is defined by a bijection $g_\mathbb{R}:\mathbb{R}\to \mathbb{R}$. It  entails a hierarchy of probabilistic models, with probabilities $p_k=g^k(p)$, where $g$ is the restriction of $g_\mathbb{R}$ to the interval $[0,1]$, $g^k$ is the $k$th iterate of $g$, and $k$ is an arbitrary integer (positive, negative, or zero; $g^0(x)=x$). The relation between $p$ and $g^k(p)$, $k>0$, is analogous to the one between probability and neural activation function.  For \mbox{$k\ll -1$}, $g^k(p)$ is essentially white noise (all processes are equally probable). The choice of $k=0$ is physically as arbitrary as the choice of origin of a line in space, 
hence what we regard as experimental binary probabilities, $p_{\rm exp}$, can be given by any $k$,  $p_{\rm exp}=g^k(p)$. 
Quantum binary probabilities are defined by $g(p)=\sin^2\frac{\pi}{2}p$. With this concrete form of $g$, one finds that any two neighboring levels of the hierarchy are related to each other in a quantum--subquantum  relation. In this sense, any model in the hierarchy is probabilistically quantum in appropriate arithmetic and calculus. And the other way around: any model is subquantum in appropriate arithmetic and calculus.
Probabilities involving more than two events are constructed by means of trees of binary conditional probabilities.  We discuss from this perspective singlet-state probabilities and Bell inequalities. We find that singlet state probabilities involve simultaneously three levels of the hierarchy: quantum, hidden, and macroscopic. As a by-product of the analysis, we discover a new (arithmetic) interpretation of the Fubini--Study geodesic distance. 
 \end{abstract}

\maketitle

\section{Introduction}

In brief, the~quantum measurement problem consists of finding a rule that correlates states of a quantum system with those of a macroscopic observer. When phrased in  probabilistic terms, the~problem is to find a consistent rule of replacing  joint probabilities, $p(a,b)$,  by~conditional probabilities, $p(a|b)$, where $a$ and $b$ represent states (or properties) of the system and the observer, respectively. 
In standard quantum mechanics the rule can be inferred from Bayes law by the following sequence of equivalences:
\be
p(a|b)=\frac{p(a,b)}{p(b)}
=
\frac{\Tr(\rho P_bP_aP_b)}{\Tr(\rho P_b)}
=
\Tr\left(\frac{P_b\rho P_b}{\Tr(P_b\rho P_b)}P_a\right)
=
\Tr(\rho_bP_a).\label{(1)}
\ee
Thus,  the~process of conditioning by the event ``$b$ has occurred'' can be represented by the ``state vector reduction'',
\be
\rho\mapsto  \rho_b=\frac{P_b\rho P_b}{\Tr(P_b\rho P_b)}.\label{(2)}
\ee
However, do we really need (\ref{(2)})? From an operational point of view, it is enough if we know the joint probability, 
\be
p(a,b)=\Tr(\rho P_bP_aP_b),
\ee
and the probability of the condition,
\be
p(b)=\Tr(\rho P_b).
\ee
Both numbers are directly related to experimental data, so (\ref{(2)}) is~redundant.

If we try to generalize the above procedure beyond quantum mechanics, various possibilities arise.
 In nonlinear quantum mechanics, for~example, once we obtain $p(a,b)$ and $p(b)$, we can {\emph{deduce}} %MDPI: Please confirm if the italics are necessary; if not, please remove them. The following highlights are the same.
 the mathematical form of an effective state vector reduction, but~it will not coincide with (\ref{(2)}), because~the sequence of transformations (\ref{(1)}) will no longer be true (cf.~\cite{CD} for the details). A~naive combination of (\ref{(2)}) with nonlinear evolution of states implies the inconsistency known as faster-than-light communication~\cite{Gisin,Gisin',Polchinski,Jordan}. Of~course, one {\it can\/} work with the projection postulate even in nonlinear quantum mechanics (eliminating the faster-than-light effect), but~the form of state vector reduction must be first derived in a consistent way from Bayes law~\cite{CD}. 
Here, consistency is the keyword. %Please check intended meaning has been retained.

The Bayes law, when written as
$p(a,b)=p(a|b)p(b)$, 
is known as the {\it {product rule}\/}. Jaynes~\cite{Jaynes} (following the ideas of Acz\'el~\cite{Aczel} and Cox~\cite{Cox}) derives the product rule from some very general desiderata of consistent and plausible reasoning but, interestingly, what one finds turns out be more general,
\be
p(a,b)=g^{-1}\Big(g\big(p(a|b)\big)g\big(p(b)\big)\Big),\label{w}
\ee
where $g$ is some monotone non-negative function (cf. Equation~(2.27) in~\cite{Jaynes}). Still, for~Jaynes,  $p(\dots)$ is not yet a probability. His intuition tells him that the probability (or, rather, a~measure of plausibility) is given by $g(p(\dots))$, so that the product rule is reconstructed in the standard form, 
\be
g\big(p(a,b)\big)=g\big(p(a|b)\big)g\big(p(b)\big).
\ee

What we will discuss later on in this paper employs a possibility that was not taken into account by Jaynes. Namely, we will treat formulas such as (\ref{w}) as a definition of a new product, $\odot$, so that
\be
p(a,b)=g^{-1}\Big(g\big(p(a|b)\big)g\big(p(b)\big)\Big)=p(a|b)\odot p(b).\label{ww}
\ee
We will also see that $g(p)$ and its higher iterates  have intriguing similarities to neural activation functions, whereas higher iterates of 
$g^{-1}(p)$ resemble a white~noise.

A new product is an element of a new arithmetic, leading us ultimately to a whole hierarchical structure of such generalized models. As~one of the conclusions, we will find that both $p$ and $g(p)$ may be treated as genuine probabilities, provided $g$ is restricted to the class discussed in detail in Section~\ref{Section2}. One of the possibilities, directly related to the measurement problem, is that $p$ are probabilities at a hidden-variable level, whereas $g(p)$ are the quantum ones. We will see that any two neighboring levels of the hierarchy are related to each other in a way that may be regarded as a form of a quantum--subquantum relationship. This will lead to the idea of {\it {relativity of quantumness}\/}.

In any such generalized and fundamental theory one is necessarily confronted with the chicken-or-egg dilemma: What was first, $p(a,b)$ and $p(b)$, or~$p(a|b)$ and $p(b)$? The Bayes law that defines the conditional probability in terms of the joint probability, or~the product rule that defines the joint probability in terms of the conditional probability? 

An alternative form of the dilemma can be expressed in terms of the projection postulate: Do we first define conditional probabilities in terms of some given form of state vector reduction, or~we begin with joint probabilities and then infer the form of  state vector reduction? In nonlinear quantum mechanics, the latter strategy is superior to the former one. However, in~the Bayesian approach to probability, one updates probabilities on the basis of prior information, so the conditional probabilities are superior to the joint~ones.

The formalism of arithmetic hierarchies discussed in the present paper clearly prefers the Bayesian approach. The~reason is in the three fundamental lemmas, which we will discuss in Section~\ref{Section2}, which are true only for binary probabilities. There is  priority in the binary coding, as~ we have to construct probabilities involving more than two events in terms of binary trees of conditional probabilities. Binary coding becomes as fundamental for probability theory as the two-spinors are fundamental for relativistic physics~\cite{PR}.

We begin in Section~\ref{Section2} by recalling the three fundamental lemmas about the functional equation $g(p)+g(1-p)=1$. 
In Section~\ref{Section3}, we construct a hierarchy of isomorphic arithmetics associated with $g(p)$. The~hierarchy of arithmetics leads to a hierarchy of probabilities introduced in Section~\ref{Section4}. A~hierarchical ordering relation, briefly discussed in Section~\ref{Section5}, will allow us to unambiguously employ symbols such as $<$ and $>$. A~family of product rules, discussed in Section~\ref{Bayes}, is employed in the problem of hidden-variables representation of singlet-state probabilities in Section~\ref{Section7}. We explain, in~particular, that one encounters here three types of arithmetic levels in a single formula for joint probabilities: quantum,  macroscopic, and~ hidden. Section~\ref{Section8} introduces some elements of hierarchical calculi, with~special emphasis on non-Newtonian integration. We make here a digression on R\'enyi's entropy which is implicitly based on a generalized arithmetic, but~does not take advantage of the possibilities inherent in generalized calculus. Section~\ref{Section9} is devoted to  local hidden-variable models of singlet-state probabilities constructed in terms of the generalized calculus. This seems to be the most controversial aspect of the formalism, as~it clearly contradicts common wisdom about Bell's theorem. 
Section~\ref{Fubini} brings us to the intriguing role played in quantum mechanics by the geodesic distance in the projective space of quantum states. A~typical discussion of the Fubini--Study metric is restricted in the literature to its geometric interpretation. Here, we reveal its unknown aspect: Its role for the arithmetic structure of quantum states. It seems that $g(p)=\sin^2\frac{\pi}{2}p$ is a fundamental bijection that determines the arithmetic of the subquantum world. In~Section~\ref{SecFinal}, we give a simple argument explaining why the effective number of distinguishable probabilistic levels of the hierarchy is finite. We also point out a possible interpretation of the hierarchy of probabilities in terms of neural activation functions. 
 At such a formal level, the~only means of relating formal probabilities to experiment is via the laws of large numbers, discussed in  Section~\ref{Section11}. 
In Section~\ref{Section12}, we return to the problem of Bell's inequalities. We depart here a little from the formalism we developed in a series of earlier papers where the same arithmetic was used at the hidden and the macroscopic levels. Our current understanding of the problem is that it is better to employ the freedom of combining different arithmetics simultaneously. We end the paper with remarks on open problems, Section~\ref{Section13}, and~certain personal perspective is given in Section~\ref{Section14}. The~{Appendix} %MDPI: We changed it to Appendix A. Please confirm.
 \ref{appa} is devoted to certain technicalities which cannot be found in the~literature.

\section{Three Fundamental~Lemmas}
\label{Section2}

The hierarchical structure of (binary) probabilities is a consequence of the following three lemmas.
They do not have a sufficiently nontrivial generalization beyond the binary case (cf. the discussion in~\cite{MCEntropy}), hence the non-binary case has to be treated in terms of trees of conditional probabilities constructed in analogy to binary {Huffman codes}~\cite{Huffman}. %MDPI: We changed Lemma format as required. Please confirm whether the lemma content are all included in the Lemma environment.
\begin{lemma}

$g:[0,1]\to [0,1]$ is a solution of the functional equation
$g(p)+g(1-p)=1$  if and only if  
\be
g(p)=\frac{1}{2} + h\left(p-\frac{1}{2}\right),\label{6}
\ee
where $h(-x)=-h(x)$, $h:[-1/2,1/2]\to[-1/2,1/2]$, i.e.,~$h$ is an {\it arbitrary\/} odd mapping of the closed interval into itself. Any such $g$ has a fixed point at $p=1/2$.
\end{lemma}
\begin{lemma}
Consider two functions $g_j:[0,1]\to [0,1]$, $j=1,2$, that satisfy assumptions of Lemma~1,
\be
g_j(p)=\frac{1}{2} + h_j\left(p-\frac{1}{2}\right),
\ee
where $h_j(-x)=-h_j(x)$. 
Then $g_{12}=g_1\circ g_2$ also satisfies Lemma~1 with $h_{12}=h_1\circ h_2$,
\be
g_{12}(p)=\frac{1}{2} + h_{12}\left(p-\frac{1}{2}\right).
\ee
Accordingly,
\be 
g_{12}(p)+g_{12}(1-p)=1
\ee
for any $p\in[0,1]$.
\end{lemma}

\begin{lemma}
Let $g^{k}=g\circ \dots \circ g$, $g^{-k}=g^{-1}\circ \dots \circ g^{-1}$ ($k$ times), $g^0(x)=x$. If~$g$ satisfies Lemma~1, 
\be
g(p)=\frac{1}{2} + h\left(p-\frac{1}{2}\right),
\ee
then the $k$th iterate  $g^{k}$ also satisfies Lemma~1 for any $k\in\mathbb{Z}$,
\be
g^{k}(p)=\frac{1}{2} + h^k\left(p-\frac{1}{2}\right),
\ee
where $h^k$ is the $k$th iterate of $h$.
Accordingly,
\be 
g^{k}(p)+g^{k}(1-p)=1\label{57a}
\ee
for any $p\in[0,1]$, and~any integer $k$. In~particular
\be 
g^{-1}(p)+g^{-1}(1-p)=1.\label{58}
\ee
The proofs can be {found in}~\cite{MCAPPA2021,MCKN}.%MDPI: incorrect references citation order, We have rearranged all references citation to make them appear in numerical order, please confirm.

\end{lemma}

Armed with the lemmas we can construct a hierarchy of arithmetics, entailing a hierarchy of~probabilities.

\section{Hierarchy of Isomorphic~Arithmetics}
\label{Section3}

Assume that $g:[0,1]\to [0,1]$ occurring in the above three lemmas is a restriction of a bijection $g_\mathbb{R}:\mathbb{R}\to \mathbb{R}$, i.e.,~$g(x)=g_\mathbb{R}(x)$ for $x\in[0,1]$. It does not matter what the properties of $g_\mathbb{R}(x)$ are if $x\not\in [0,1]$, except~for the bijectivity of  
$g_\mathbb{R}$. Put differently, $g$ belongs to the equivalence class $[g_\mathbb{R}]$ of bijections whose restrictions to $[0,1]$ are identical. Following the notation of Lemma~3, we denote $g^{k}=g_\mathbb{R}\circ \dots \circ g_\mathbb{R}$, 
$g^{-k}=g^{-1}_\mathbb{R}\circ \dots \circ g^{-1}_\mathbb{R}$, $g^0(x)=x$. Now, let $x,y\in\mathbb{R}$. Define,
\be
x\oplus_k y  &=& g^{k}\Big( g^{-k}(x)+g^{-k}(y)\Big),\label{12}\\
x\ominus_k y  &=& g^{k}\Big( g^{-k}(x)-g^{-k}(y)\Big),\label{13}\\
x\odot_k y  &=& g^{k}\Big( g^{-k}(x)\cdot g^{-k}(y)\Big),\label{14}\\
x\oslash_k y  &=& g^{k}\Big( g^{-k}(x)/g^{-k}(y)\Big).\label{15}
\ee
The arithmetic $\mathbb{R}_k$ is the set $\mathbb{R}$ equipped with the above four operations, i.e., \linebreak  \mbox{$\mathbb{R}_k=\{\mathbb{R},\oplus_k,\ominus_k,\odot_k,\oslash_k\}$}. The~ordering relation is independent of $k$ if $g$ is increasing, which we therefore assume, hence $g^{k}(x)<g^{k}(y)$ if and only if $x<y$. 
The neutral elements of addition, $0_k=g^{k}(0)$, and~multiplication, $1_k=g^{k}(1)$,
\be
x\oplus_k 0_k  =x\odot_k 1_k  = x,\quad \textrm{for any $x$,}
\ee
can be regarded as bits, in~principle applicable to some form of binary coding.
Greater natural numbers are obtained by the $n$-times repeated addition of $1_k$,
\be
n_k
&=&
\underbrace{1_k\oplus_k\dots\oplus_k 1_k}_{\textrm{$n$ times}}=g^{k}(n),\\
n_k\oplus_k m_k
&=&
g^{k}(n+m)=(n+m)_k,\\
n_k\odot_k m_k
&=&
g^{k}(nm)=(nm)_k.
\ee
An $n$th power of $x$,
\be
x^{n_k}
&=&
\underbrace{x\odot_k\dots\odot_k x}_{\textrm{$n$ times}},
\ee
satisfies
\be
x^{n_k}\odot_k x^{m_k}=x^{(n+m)_k}=
x^{n_k\oplus_k m_k}.
\ee
Rational numbers are those of the form
\be
n_k\oslash_k m_k=g^{k}(n/m)=(n/m)_k,\quad \textrm{$n,m\in\mathbb{Z}$.}\label{17''}
\ee
The notion of rationality is arithmetic-dependent. Indeed, let $n/m$ be a rational number in the arithmetic $\mathbb{R}_0=\{\mathbb{R},+,-,\cdot,/\}$. Then, typically,
$g^{k}(n/m)$, $k\neq 0$, is not a rational number in $\mathbb{R}_0$. Still, it is a rational number in the arithmetic $\mathbb{R}_k=\{\mathbb{R},\oplus_k,\ominus_k,\odot_k,\oslash_k\}$ in consequence of (\ref{17''}). 

For any $k,l\in\mathbb{Z}$, the~four arithmetic operations  are related by
\be
x\odot_{k+l} y
&=&
g^{l}\Big(
g^{-l}(x)\odot_k g^{-l}(y)\Big)
=
g^{k}\Big(
g^{-k}(x)\odot_l g^{-k}(y)\Big),\label{16}\\
x\oslash_{k+l} y
&=&
g^{l}\Big(
g^{-l}(x)\oslash_k g^{-l}(y)\Big)
=
g^{k}\Big(
g^{-k}(x)\oslash_l g^{-k}(y)\Big),\label{17}\\
x\oplus_{k+l} y
&=&
g^{l}\Big(
g^{-l}(x)\oplus_k g^{-l}(y)\Big)
=
g^{k}\Big(
g^{-k}(x)\oplus_l g^{-k}(y)\Big),\label{18}\\
x\ominus_{k+l} y
&=&
g^{l}\Big(
g^{-l}(x)\ominus_k g^{-l}(y)\Big)
=
g^{k}\Big(
g^{-k}(x)\ominus_l g^{-k}(y)\Big).\label{19}
\ee
The bijection $f^k=g^{-k}$  is an isomorphism of $\mathbb{R}_{k+l}$ and $\mathbb{R}_l$, for~any $k,l\in\mathbb{Z}$, 
\be
f^{k}
\left(
x\odot_{k+l} y
\right)
&=&
f^{k}(x)\odot_l f^{k}(y),\\
f^{k}
\left(
x\oslash_{k+l} y
\right)
&=&
f^{k}(x)\oslash_l f^{k}(y),\\
f^{k}
\left(
x\oplus_{k+l} y
\right)
&=&
f^{k}(x)\oplus_l f^{k}(y),\\
f^{k}
\left(
x\ominus_{k+l} y
\right)
&=&
f^{k}(x)\ominus_l f^{k}(y).
\ee
The value $l=0$ is not privileged. The~role of a 0th level can be played by any $l$. The~notation~where
\be
\mathbb{R}_l=\{\mathbb{R},\oplus_l, \ominus_l, \odot_l, \oslash_l\}=
\{\mathbb{R},+, -, \cdot, /\},
\ee
is perfectly acceptable, hence  {\it any\/} $\mathbb{R}_l$ can be regarded as ``the'' ordinary arithmetic we are taught at  school. The~latter statement is the content of the ``arithmetic Copernican principle'', {introduced in}%MDPI: incorrect references citation order, We have rearranged all references citation to make them appear in numerical order, please confirm.
~\cite{MCKN} and discussed further in~\cite{MCAPPA2023}. In~the present paper we nevertheless simplify notation and assume $\mathbb{R}_0=\{\mathbb{R},+, -, \cdot, /\}$. This is analogous to the usual habit of imposing initial conditions in Newtonian dynamics  ``at $t=0$'' instead of a general $t=t_0$. 

The hierarchy of arithmetics leads to the hierarchy of~probabilities.

\section{Hierarchy of~Probabilities}
\label{Section4}

Let $g(1)=1$, so that $1_k=g^{k}(1)=1$ and $0_k=g^{k}(0)=0$, for~any $k$. Now, let $p$, $q$, $p+q=1$, be probabilities.  Assuming that $g$ satisfies the assumptions of  Lemma~1, we find (in consequence of Lemmas 2 and 3, and~$g^{k}(1)=1$ for any $k\in\mathbb{Z}$)
\be
p+q &=&1,\label{28}\\
g^{k}(p)+g^{k}(q) &=&1,\label{29}\\
p\oplus_{-k} q=g^{-k}\left(g^{k}(p)+g^{k}(q)\right) &=&1 ,\label{29''}
\ee
for any $k\in\mathbb{Z}$. The~Copernican aspect is visible at the level of probabilities as well,  if~we define $P=g^{k}(p)$, $Q=g^{k}(q)$, so that
\be
g^{-k}(P)+g^{-k}(Q) &=&1,\label{31}\\
P+Q &=&1,\label{30}\\
P\oplus_{k}Q=g^{k}\left(g^{-k}(P)+g^{-k}(Q)\right) &=&1,\label{31''}
\ee
for any $k\in\mathbb{Z}$. 
Indeed, how to distinguish between (\ref{28})--(\ref{29''}) and (\ref{31})--(\ref{31''}), if~we bear in mind that $k$ can be positive, negative, or~zero, and~the formulas are true for all $k$? How to distinguish between the two levels if in both cases we find $p+q=1$ and $P+Q=1$? 
Which of the probabilities, $p$ or $P$, is the one we measure in experiment? Which iterate, $k$, 0, or~$-k$, is the one that defines {\it our\/} probabilities we experimentally define in terms of frequencies of successes? Which natural numbers $n_k$, $n=n_0$, or~$n_{-k}$, are the ones we use to define numbers of trials and successes?

Formula  (\ref{29''})  shows that probabilities $p$ and $q$ sum to 1 in infinitely many ways, corresponding to infinitely many values of $k$ in $\oplus_{-k}$. Formula 
(\ref{29})  shows that  probabilities $p$ and $q$  generate infinitely many probabilities $p_k=g^{k}(p)$ and  $q_k=g^{k}(q)$ that sum to 1 by means of the same addition $+=\oplus_0$. The Arithmetic Copernican Principle is a relativity principle which states that any value of $k$ can correspond to the arithmetic and probability that we regard as ``the human and experimental~one''.

Still, this is not the end of the story. Replacing in (\ref{29}) $k$ by $k-l$,
\be
g^{k-l}(p)+g^{k-l}(q) &=&1,
\ee
and acting on both sides with $g^l$, we find
\be
g^l\left(g^{k-l}(p)+g^{k-l}(q)\right) =g^k(p)\oplus_l g^k(q)=1,\label{34''}
\ee
for any $k,l\in\mathbb{Z}$. The~resulting wealth of available probability models implied by a single bijection $g$ is truly overwhelming, yet ignored by those who study quantum probabilities and the hidden variables~problem.

Let us now consider the concrete case of the equivalence class of a function $g_\mathbb{R}$ whose restriction to $[0,1]$ is given by $g(x)=\sin^2\frac{\pi}{2}x$. Then,
\be
h(x)
&=&
g\left(x+\frac{1}{2}\right)-\frac{1}{2}
=
\frac{1}{2}\sin \pi x,\quad -\frac{1}{2}\le x\le \frac{1}{2},\\
g(p)
&=&
\frac{1}{2}+h\left(p-\frac{1}{2}\right)
=
\frac{1}{2}+\frac{1}{2}\sin \pi \left(p-\frac{1}{2}\right),\quad 0\le p\le 1,\label{36}
\ee
Let $p=(\pi-\theta)/\pi$ be the probability of finding a point belonging to the overlap of two half-circles rotated by 
$\theta\in[0,\pi]$. Then, for~$k=1$, $q=\theta/\pi$,
\be
P&=&g(p)=g^{k}(p)=\sin^2\frac{\pi}{2}\frac{\pi-\theta}{\pi}=\cos^2\frac{\theta}{2},\label{35''}\\
Q&=&g(q)=g^{k}(q)=\sin^2\frac{\pi}{2}\frac{\theta}{\pi}=\sin^2\frac{\theta}{2},\label{36''}
\ee
in which we recognize the conditional probabilities for two successive measurements of spin-1/2 in two Stern--Gerlach devices placed one after another, with~relative angle $\theta$. 

By Lemma~3, we have in fact  much more, because~$k=1$ can be replaced by any integer. For~example, the~second iterate
\be
P=g^2(p)=g\big(g(p)\big)=\sin^2\frac{\pi}{2}\left(\cos^2\frac{\theta}{2}\right),
\ee
satisfies $g^2(p)+g^2(q)=1$, of~course, as~can be proved by a straightforward but instructive calculation~\cite{MCAPPA2023}. The~minus-first iterate,
\be
P=g^{-1}(p)=\frac{2}{\pi}\arcsin\sqrt{p}
=
\frac{2}{\pi}\arcsin\sqrt{\frac{\pi-\theta}{\pi}},\label{39}
\ee
satisfies $g^{-1}(p)+g^{-1}(q)=1$, and~so on and so~forth.

Clearly, we have absolutely no criterion that could indicate which level of the hierarchy is the one we regard as our human one, a~fact that justifies the adjective ``Copernican''. For~example, rewriting (\ref{39}) as
\be
P&=&g^{1-2}(p)=g^{1}\left(g^{-2}(p)\right)=g^{1}\left(1-g^{-2}(q)\right)=g^{1}\left(1-\frac{\alpha}{\pi}\right)
=
\cos^2\frac{\alpha}{2},
\ee
we find the relation between the two parameters, $\alpha$ and $\theta$, corresponding to the two levels of the hierarchy (see Figure~\ref{Fig1}),
\be
\alpha(\theta)
=
\pi g^{-2}(q)
=
2\arcsin\sqrt{\frac{2}{\pi}\arcsin\sqrt{\frac{\theta}{\pi}}}\label{41}.
\ee

\vspace{-6pt}
\begin{figure}

\includegraphics[width=8 cm]{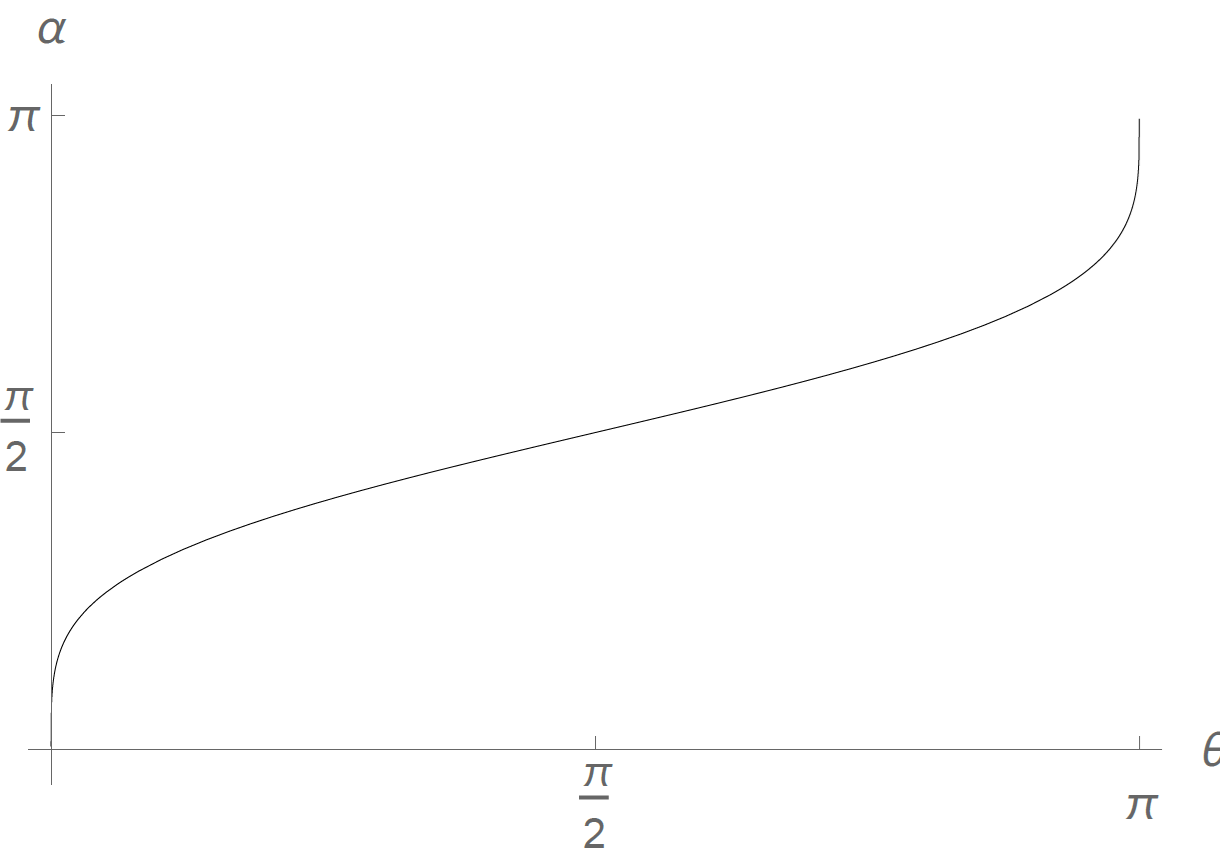}
\caption{The relation between $\alpha$ and $\theta$ as given by (\ref{41}). There are three fixed points: $\alpha(0)=0$, $\alpha(\pi/2)=\pi/2$, $\alpha(\pi)=\pi$. Here, $\alpha$ is the geometric angle between the two Stern--Gerlach devices, whereas $\theta$ is a hidden~parameter.}
\label{Fig1}
\end{figure}   

The usual tests of classicality and quantumness are based on inequalities. However, in~order to discuss an inequality we have to control ordering  relations such as $\le$ and $\ge$. Fortunately, with~our assumptions about $g$ the problem is~trivial.

\section{Hierarchical Ordering~Relation}
\label{Section5}

We assume that the bijection $g$ is strictly increasing, i.e.,~$x<y$ if and only if $g(x)<g(y)$. A~composition of two strictly increasing functions is strictly increasing, hence   
$x<y$ implies $g^k(x)\ominus_l g^k(y)<0_l=0$ for any $k,l\in\mathbb{Z}$. The~latter leads to a unique ordering relation at the level of the entire hierarchy of arithmetics. This is why it is safe to use the symbols $<$, $>$, $\le$, $\ge$ at any level of the~hierarchy.

So far, we have restricted our analysis to binary events. An~extension to higher dimensional problems needs the notion of a product~rule.

\section{Hierarchical Product~Rules}
\label{Bayes}

The standard  product rule states that probability of a sequence of two events, first $a_1$ then $a_2$,  is given by the product of the prior $p(a_1)$ (a probability of the condition) with the posterior $p(a_2|a_1)$ (a conditional probability of $a_2$ under the condition that $a_1$ has happened). The~sums of binary probabilities,
\be
g^{k_1}\big(p(0)\big)\oplus_l g^{k_1}\big(p(1)\big) &=& 1, \quad \textrm{for any $k_1,l\in\mathbb{Z}$},\\
g^{k_2}\big(p(0|a_1)\big)\oplus_l g^{k_2}\big(p(1|a_1)\big) &=& 1, \quad \textrm{for any $k_2,l\in\mathbb{Z}$},
\ee
as implied by the lemmas, are naturally related to
\be
g^{k_2}\big(p(a_2|a_1)\big)\odot_l g^{k_1}\big(p(a_1)\big), \quad \textrm{for any $k_1,k_2,l\in\mathbb{Z}$},\label{42}
\ee
because
\be
\bigoplus_{a_1,a_2}{\!_l}\,g^{k_2}\big(p(a_2|a_1)\big)\odot_l g^{k_1}\big(p(a_1)\big)=1, \quad \textrm{for any $k_1,k_2,l\in\mathbb{Z}$}.\label{58''}
\ee
A sequence of results, $a_n, a_{n-1},\dots, a_1$, implies their joint probability
\be
g^{k_n}\big(p(a_n|a_{n-1}\dots a_1)\big)\odot_l\dots \odot_l g^{k_2}\big(p(a_2|a_1)\big)\odot_l g^{k_1}\big(p(a_1)\big)
\label{59''}
\ee
normalized by
\be
\bigoplus_{a_1\dots a_n}{\!\!\!_l}\,g^{k_n}\big(p(a_n|a_{n-1}\dots a_1)\big)\odot_l\dots \odot_l g^{k_2}\big(p(a_2|a_1)\big)\odot_l g^{k_1}\big(p(a_1)\big)
=g^l(1)=1.\label{60''}
\ee
In particular, for~$l=0$, 
\be
g^{k_1}\big(p(0)\big)+ g^{k_1}\big(p(1)\big) &=& 1, \quad \textrm{for any $k_1\in\mathbb{Z}$},\\
g^{k_2}\big(p(0|a_1)\big)+ g^{k_2}\big(p(1|a_1)\big) &=& 1, \quad \textrm{for any $k_2\in\mathbb{Z}$},
\ee
and
\be
\sum_{a_1,a_2}g^{k_2}\big(p(a_2|a_1)\big)g^{k_1}\big(p(a_1)\big)=1, \quad \textrm{for any $k_1,k_2\in\mathbb{Z}$}.\label{71''''}
\ee
At the other extreme is the case of $l=k_1=k_2=k$,
\be
g^{k}\big(p(a_2|a_1)\big)\odot_k g^{k}\big(p(a_1)\big)
=
g^{k}\big(p(a_2|a_1)p(a_1)\big),
\ee
with normalization
\be
\bigoplus_{a_1,a_2}{\!\!_k}\,g^{k}\big(p(a_2|a_1)p(a_1)\big)=g^{k}\left(\sum_{a_1,a_2}p(a_2|a_1)p(a_1)\right)=1, \quad \textrm{for any $k\in\mathbb{Z}$}.
\ee
It is striking that in formulas such as (\ref{59''})  each of the $k$-indices can be in principle different. In~effect, (\ref{59''}) may be regarded as a component of a~tensor.

A truly nontrivial application of generalized product rules occurs in the problem of singlet-state probabilities, quantum entangled states, and~Bell's~theorem. 

\section{Singlet-State~Probabilities}
\label{Section7}

Singlet-state probabilities occur in experiments where two parties (``Alice'' and ``Bob'') are macroscopically separated, but~the measurements they perform are the quantum ones.  Such probabilities naturally occur in the context of the hierarchical product rule. Indeed, consider the following probabilities,
\be
p(0) &=& p(1)=\frac{1}{2},\label{53}\\
p(0|0) &=& p(1|1)=\frac{\theta}{\pi},\label{53_}\\
p(1|0) &=& p(0|1)=\frac{\pi-\theta}{\pi}\label{53__},
\ee
whose geometric  interpretation is evident. As~the bijection take the one occurring \mbox{in (\ref{36})--(\ref{36''})}. Then,
\be
g\big(p(0|0)\big)g\big(p(0)\big)
&=&
g\big(p(1|1)\big)g\big(p(1)\big)
=
\frac{1}{2}\sin^2\frac{\theta}{2},\\
g\big(p(1|0)\big)g\big(p(0)\big)
&=&
g\big(p(0|1)\big)g\big(p(1)\big)
=
\frac{1}{2}\cos^2\frac{\theta}{2},\label{57}
\ee
are the probabilities typical of the singlet state.
Let us note that we have employed the product rule,
\be
g^{k_2}\big(p(a|b)\big)\odot_l g^{k_1}\big(p(b)\big)
=
g^1\big(p(a|b)\big)\odot_0 g^{k_1}\big(p(b)\big),
\ee
with $k_2\neq l$. $k_1$ can be arbitrary because $g(1/2)=1/2=g^{k_1}(1/2)$ for any $g$ that satisfies Lemma~1. For~simplicity, we set $k_1=1$. Now, the~joint probability can be interpreted as~follows:
\be
P(a,b)=\underbrace{g\big(\overbrace{p(a|b)}^{\textrm{hidden}}\big)}_{\textrm{quantum}}\underbrace{\odot_0}_{\textrm{macroscopic}} \underbrace{g\big(\overbrace{p(b)}^{\textrm{hidden}}\big)}_{\textrm{quantum}}.
\label{P(a,b)}
\ee
Let us further note that we could have started with the following:
\be
g^{k}\big(p(0)\big) &=& g^{k}\big(p(1)\big)=\frac{1}{2},\\
g^{k}\big(p(0|0)\big) &=& g^{k}\big(p(1|1)\big)=\frac{\theta}{\pi},\\
g^{k}\big(p(0|1)\big) &=& g^{k}\big(p(1|0)\big)=\frac{\pi-\theta}{\pi}.
\ee
Then, $g^{k+1}\big(p(a_2|a_1)\big)g^{k+1}\big(p(a_1)\big)$ would be the singlet-state~probabilities.

One concludes that the notion of a quantum level is a relative one. In~fact, any level is quantum, and~any level is hidden; moreover, any $\odot_l$ can play the role of the macroscopic arithmetic. What counts is the neighboring location in the hierarchy. 
The so-called violation of Bell's inequality is an inconsistency that occurs if we apply the arithmetic of a hidden level to calculations performed at the neighboring quantum one. An~analogous inconsistency that occurs between non-neighboring levels leads to violations beyond the Tsirelson bound~\cite{Tsirelson,MCKN}. 

In order to perform calculations at different levels of the hierarchy, we have to understand what  the consequences  are of the hierarchical structure of arithmetics for the resulting hierarchy of~calculi.

\section{Hierarchy of~Calculi}
\label{Section8}

A hierarchy of arithmetics leads to a hierarchy of ``non-Newtonian'' calculi \citep{GK,G79,G79',P,P',GMMP}. Here, functions such as $A:\mathbb{R}\to \mathbb{R}$ have to be treated as mappings between arithmetics and not between sets, hence it is more appropriate to write 
\be
A_{lk}:\mathbb{R}_k\to \mathbb{R}_l,
\ee
with some $k,l\in\mathbb{Z}$. Otherwise the notions of derivative and integral are ambiguous. The~derivative of $A_{lk}$ is 
\be
\frac{{\rm D}_l A_{lk}(x)}{{\rm D}_k x}
&=&
\lim_{\delta\to 0}
\Big(
A_{lk}(x\oplus_k\delta_k)\ominus_l A_{lk}(x)
\Big)
\oslash_l \delta_l.\label{D/Dx}
\ee
As before, $\delta_k=g^k(\delta)$, $\delta_l=g^l(\delta)$. The~derivative is $\mathbb{R}_l$-linear and satisfies 
an appropriate Leibniz rule,
\be
\frac{{\rm D}_l \big(A_{lk}(x)\oplus_l B_{lk}(x)\big)}{{\rm D}_k x}
&=&
\frac{{\rm D}_l A_{lk}(x)}{{\rm D}_k x}
\oplus_l
\frac{{\rm D}_l B_{lk}(x)}{{\rm D}_k x},\\
\frac{{\rm D}_l \big(A_{lk}(x)\odot_l B_{lk}(x)\big)}{{\rm D}_k x}
&=&
\left(
\frac{{\rm D}_l A_{lk}(x)}{{\rm D}_k x}\odot_l B_{lk}(x)
\right)
\oplus_l
\left(
A_{lk}(x)\odot_l
\frac{{\rm D}_l B_{lk}(x)}{{\rm D}_k x}
\right).
\ee
Integration of $A_{lk}:\mathbb{R}_k\to \mathbb{R}_l$ is defined in a way that guarantees the two fundamental theorems of calculus (under standard assumptions about differentiability and  continuity):
\be
\int_a^b
\frac{{\rm D}_l A_{lk}(x)}{{\rm D}_k x} {\rm D}_k x
&=&
A_{lk}(b)\ominus_l A_{lk}(a),\\
\frac{{\rm D}_l }{{\rm D}_k x}
\int_a^x
A_{lk}(y) {\rm D}_k y
&=&
A_{lk}(x).
\ee
The formulas become less abstract if one considers the following commutative diagram ($f=g^{-1}$)
\be
\begin{array}{rcl}
\mathbb{R}_k                & \stackrel{A_{lk}}{\longrightarrow}       & \mathbb{R}_l               \\
f^k{\Big\downarrow}    &                                     & {\Big\uparrow}g^l  \\
\mathbb{R}_0                & \stackrel{A_{00}}{\longrightarrow}   & \mathbb{R}_0\\
g^{n}{\Big\downarrow}    &                                     & {\Big\uparrow}f^{m}  \\
\mathbb{R}_n                & \stackrel{A_{mn}}{\longrightarrow}       & \mathbb{R}_m               
\end{array},
\label{diagramf}
\ee
leading to a very simple and useful form of the derivative (\ref{D/Dx}),
\be
\frac{{\rm D}_l A_{lk}(x)}{{\rm D}_k x}
&=&
g^l
\left(
\frac{{\rm d}A_{00}\big(f^k(x)\big)}{{\rm d} f^k(x)}
\right),\label{derivative''}
\ee
while the integral reads,
\be
\int_a^b
A_{lk}(x) {\rm D}_k x
&=&
g^l
\left(
\int_{f^k(a)}^{f^k(b)}
A_{00}(r) {\rm d}r
\right).\label{integr}
\ee
Here, ${\rm d}r$ denotes the usual (Riemann, Lebesgue, etc.) integral in $\mathbb{R}_0$.
Formula (\ref{derivative''}) is derived under the assumption that $g:\mathbb{R}\to \mathbb{R}$ is continuous (in the usual meaning of the term employed in ordinary ``Newtonian'' real analysis), which is however, automatically guaranteed by the fact that $g$ is a bijection. What is important, neither $g$ nor its inverse $f$ have to be differentiable in the standard Newtonian sense. The~latter makes an important difference with respect to the ordinary differential geometry where functions such as $g(x)=x^{1/3}$ would be excluded as non-differentiable at $x=0$. In~the non-Newtonian formalism, any bijection $g$, as~well as its inverse $f$, are automatically smooth with respect to the non-Newtonian differentiation defined by the same $g$. Various explicit examples can be found in~\cite{BC,Czachor2019}.

Linearity of the integral must be understood in the sense of $\mathbb{R}_l$,
\be
\int_a^b
A_{lk}(x)\oplus_l B_{lk}(x) {\rm D}_k x
&=&
\int_a^b
A_{lk}(x){\rm D}_k x
\oplus_l 
\int_a^b
 B_{lk}(x) {\rm D}_k x,\\
\int_a^b
A_{l}\odot_l B_{lk}(x) {\rm D}_k x
&=&
A_{l}
\odot_l 
\int_a^b
 B_{lk}(x) {\rm D}_k x, \quad \textrm{for a constant $A_{l}\in\mathbb{R}_l$},
\ee
a property of fundamental importance for Bell-type inequalities~\cite{MCKN}. An~analogous form of generalized linearity of integrals occurs in fuzzy calculus~\cite{fuzzy,fuzzy1,fuzzy1',fuzzy1'',Pap2020}.

Diagram (\ref{diagramf}) implies
\be
A_{lk} =g^l\circ A_{00}\circ f^k=g^{l-m}\circ g^m\circ A_{00}\circ f^n\circ f^{k-n}
=g^{l-m}\circ  A_{mn}\circ f^{k-n},
\ee
which leads to a new type of a chain rule, relating derivatives and integrals at different levels of the hierarchy,
\be
\frac{{\rm D}_l A_{lk}(x)}{{\rm D}_k x}
&=&
g^{l-m}
\left(
\frac{{\rm D}_m A_{mn}\left(f^{k-n}(x)\right)}{{\rm D}_n f^{k-n}(x)}
\right),\label{86''}\\
\int_a^b
A_{lk}(x) {\rm D}_k x
&=&
g^{l-m}
\left(
\int_{f^{k-n}(a)}^{f^{k-n}(b)}
A_{mn}(x) {\rm D}_nx
\right).\label{87''}
\ee
Formulas  (\ref{86''}) and (\ref{87''}) do not seem to appear in the literature, so we prove them in \mbox{Appendix~\ref{appa}}.

\subsection*{{Digression: Logarithm and R\'enyi~Entropies}} %MDPI: If there is only one subsection within a section, it should not be numbered. We have thus removed this section number. Please confirm.
\label{Section8.1}

Exponential function is defined by the differential equation,
\be
\frac{{\rm D}_l \exp_{lk}(x)}{{\rm D}_k x}
&=&
g^l
\left(
\frac{{\rm d}\exp_{00}\big(f^k(x)\big)}{{\rm d} f^k(x)}
\right)=\exp_{lk}(x)
=
g^l
\left(
\exp_{00}\big(f^k(x)\big)
\right),\label{exp}\\
\exp_{lk}(0_k)
&=&
1_l.
\ee
The solution is given by $\exp_{00}(x)=e^x$ and satisfies 
\be
\exp_{lk}(x\oplus_k y)
=
\exp_{lk}(x)
\odot_l
\exp_{lk}(y).
\ee
The inverse is given by
\be
\ln_{kl}(x)
=
g^k
\left(
\ln_{00}\big(f^l(x)\big)
\right),\label{ln}
\ee
where $\ln_{00}(x)=\ln x$, and~\be
\ln_{kl}(x\odot_l y)
=
\ln_{kl}(x)
\oplus_k
\ln_{kl}(y).
\ee
Now, consider $\phi_\alpha (x)=e^{(1-\alpha)x}$, $\phi_\alpha ^{-1}(x)=\frac{1}{1-\alpha}\ln x$.
R\'enyi introduced his $\alpha$-entropy as a Kolmogorov--Nagumo average~\cite{R,K1930,N1930,JA,JA',CN,N2013,JK2020} of the Shannon amount of information~\cite{Shannon} (we prefer the natural logarithm to the original $\log_2$ from~\cite{R}, but~this is just a choice of units of information),
\be
S_\alpha 
&=&
\phi_\alpha ^{-1}
\left(
\sum_p  p\phi_\alpha (-\ln p)
\right)
=
\frac{1}{1-\alpha}\ln \left(\sum_p  p^\alpha\right).
\label{Renyi}
\ee
It is clear that (\ref{Renyi}) can be expressed in several different ways by means of generalized arithmetics.
For example, 
\be
\ominus_{1}
\ln_{1,0}(x)
=
g^{1}
\left(
-\ln\big(f^0(x)\big)
\right)=
g
\left(
-\ln x
\right),%\label{ln}
\ee
has the same functional form as $\phi_\alpha \big(-\ln p\big)$. Alternatively, defining
\be
x\oplus y
&=&
\phi_\alpha ^{-1}
\left(
\phi_\alpha (x)+\phi_\alpha (y)
\right),\\
x\odot y
&=&
\phi_\alpha ^{-1}
\left(
\phi_\alpha (x)\phi_\alpha (y)
\right),
\ee
and $\phi_\alpha ^{-1}(p)=P$, we find
\be
S_\alpha 
=
\phi_\alpha ^{-1}
\left(
\sum_P  \phi_\alpha (P)\phi_\alpha \big(-\ln \phi_\alpha (P)\big)
\right)
=
\bigoplus_PP\odot \ln \big(1/\phi_\alpha (P)\big).
\ee
R\'enyi's choice of $\phi_\alpha (x)=e^{(1-\alpha)x}$ was dictated by the assumed additivity of entropy for independent (i.e., uncorrelated) systems. Our general formalism suggests various hierarchical generalizations of the notion of entropy, automatically inheriting the additivity properties from the arithmetics involved.
Some examples can be found in~\cite{MCEntropy}.

\section{Application: Local Hidden-Variable Models Based on Non-Newtonian Integration}
\label{Section9}

Consider an integral representation of the standard $\mathbb{R}_0$-valued probability, with~probability densities $\rho_{00}$ and characteristic functions
\be
\chi_{\varphi, 00}(\lambda)
=
\left\{
\begin{array}{cl}
1 & \textrm{if $\lambda\in[\varphi-\pi/2,\varphi+\pi/2]$}\\
0 & \textrm{if $\lambda\not\in[\varphi-\pi/2,\varphi+\pi/2]$}
\end{array}
\right.
\ee
treated as mappings $\mathbb{R}_0\to \mathbb{R}_0$. For~example, setting  $\theta=\alpha-\beta$ in  (\ref{53}) and (\ref{53_}) one can express the probabilities in integral forms,
\be
\frac{1}{2}
&=&
\int\chi_{\alpha,00}(\lambda)\rho_{00}(\lambda)\,{\rm d}\lambda 
=
\frac{1}{2\pi}\int_{\alpha-\pi/2}^{\alpha+\pi/2}{\rm d}\lambda,\label{75''}\\
\frac{1}{2}\frac{\alpha-\beta}{\pi}
&=&
\int\chi_{\alpha,00}(\lambda)\chi_{\beta+\pi,00}(\lambda)\rho_{00}(\lambda)\,{\rm d}\lambda 
=
\frac{1}{2\pi}\int_{\beta+\pi/2}^{\alpha+\pi/2}{\rm d}\lambda,\quad \beta\le\alpha.\label{76''}
\ee
$\chi_{\varphi, 00}(\lambda)$ is the characteristic function of the half-circle located symmetrically with respect to the angle $\varphi$; $\rho_{00}(\lambda)=1/(2\pi)$ is the uniform probability density on the circle. Formula~(\ref{76''}) is {\it {local}%MDPI: Please confirm if the italics are necessary; if not, please remove them. The following highlights are the same.
\/} in the sense of Bell~\cite{Bell} and Clauser and Horne~\cite{CH}, because~of the {\it {product structure}\/} of the term
\be
\chi_{\alpha,00}(\lambda)\chi_{\beta+\pi,00}(\lambda)
=
\chi_{\alpha,00}(\lambda)\odot_0\chi_{\beta+\pi,00}(\lambda).
\ee
The case $k=l=0$ of Bayes law discussed in Section~\ref{Bayes} is (with $\theta=\alpha-\beta$)
\be
\frac{\alpha-\beta}{\pi}
&=&
\frac{\int\chi_{\beta+\pi,00}(\lambda)\chi_{\alpha,00}(\lambda)\rho_{00}(\lambda)\,{\rm d}\lambda}
{\int\chi_{\alpha,00}(\lambda')\rho_{00}(\lambda')\,{\rm d}\lambda' }=\frac{p(0_2, 0_1)}{p(0_1)}=\frac{p(1_2,1_1)}{p(1_1)}
\label{108''''}\\
&=&
\int \chi_{\beta+\pi,00}(\lambda)\frac{\chi_{\alpha,00}(\lambda)\rho_{00}(\lambda)}
{\int\chi_{\alpha,00}(\lambda')\rho_{00}(\lambda')\,{\rm d}\lambda' }{\rm d}\lambda,
\ee
which is equivalent to the assumption that the first measurement reduces the probability density according to
\be
\rho_{00}(\lambda)
\mapsto
\frac{\chi_{\alpha,00}(\lambda)\rho_{00}(\lambda)}
{\int\chi_{\alpha,00}(\lambda')\rho_{00}(\lambda')\,{\rm d}\lambda' }.\label{79}
\ee
Equation (\ref{79}) is an example of a classical projection postulate in theories based on $\mathbb{R}_0$ arithmetic.

Returning to the singlet case, corresponding to $k=1$, $l=0$, we can write it in analogy to (\ref{75''}) and (\ref{76''}),
{\small 
\be
g\big(p(a_2|a_1)\big)g\big(p(a_1)\big)
&=&
g\left(
\frac{\int\chi_{a_1,00}(\lambda)\chi_{a_2,00}(\lambda)\rho_{00}(\lambda)\,{\rm d}\lambda}
{\int\chi_{a_1,00}(\lambda)\rho_{00}(\lambda)\,{\rm d}\lambda }
\right)
g\left(
\int\chi_{a_1,00}(\lambda)\rho_{00}(\lambda)\,{\rm d}\lambda
\right)
\label{132}
\\
&=&
\frac{1}{2}
g\left(
2\int\chi_{a_1,00}(\lambda)\chi_{a_2,00}(\lambda)\rho_{00}(\lambda)\,{\rm d}\lambda
\right)\\
&=&
G\left(
\int\chi_{a_1,00}(\lambda)\chi_{a_2,00}(\lambda)\rho_{00}(\lambda)\,{\rm d}\lambda
\right)\label{92''_}\\
&=&
G\left(
\int\chi_{a_1\wedge a_2,00}(\lambda)\rho_{00}(\lambda)\,{\rm d}\lambda
\right)\label{92''},
\ee}
where $G(x)=\frac{1}{2}g(2x)$, and~
\be
\chi_{a_1\wedge a_2,00}(\lambda)
=
\chi_{a_1,00}(\lambda)\chi_{a_2,00}(\lambda)
\ee
is the characteristic function representing the conjunction ``$a_1$ and $a_2$''.  Notice that (\ref{92''_}) is a non-Newtonian integral 
\be
G\left(
\int\chi_{a_1,00}(\lambda)\chi_{a_2,00}(\lambda)\rho_{00}(\lambda)\,{\rm d}\lambda
\right)
=
\int\chi_{a_1,11}(\lambda)\odot_1\chi_{a_2,11}(\lambda)\odot_1\rho_{11}(\lambda)\,{\rm D}_1\lambda,
\label{94''}
\ee
of the function
\be
\chi_{a_1,11}\odot_1\chi_{a_2,11}\odot_1\rho_{11}:\mathbb{R}_1\to \mathbb{R}_1,
\ee
where
\be
\begin{array}{rcl}
\mathbb{R}_1                & \stackrel{\chi_{a_1,11}}{\longrightarrow}       & \mathbb{R}_1              \\
G^{-1}{\Big\downarrow}    &                                     & {\Big\uparrow}G  \\
\mathbb{R}_0                & \stackrel{\chi_{a_1,00}}{\longrightarrow}   & \mathbb{R}_0
\end{array},
\quad
\begin{array}{rcl}
\mathbb{R}_1                & \stackrel{\rho_{11}}{\longrightarrow}       & \mathbb{R}_1              \\
G^{-1}{\Big\downarrow}    &                                     & {\Big\uparrow}G  \\
\mathbb{R}_0                & \stackrel{\rho_{00}}{\longrightarrow}   & \mathbb{R}_0
\end{array},
\ee
and the multiplication is given by
\be
x\odot_1 y=G\big(G^{-1}(x)\odot_0 G^{-1}(y)\big)=G\big(G^{-1}(x) G^{-1}(y)\big).
\ee

The right-hand side of (\ref{94''}) has again the Bell--Clauser--Horne {{product form}\/}, the~only difference being that instead of $\odot_0$ one employs $\odot_1$. This is why  (\ref{94''}) can be regarded as a local hidden-variable representation of singlet-state probabilities, hence a counterexample to Bell's theorem. 
This is the main idea of the approach to singlet-state correlations introduced in~\cite{MCAPPA2021} and further discussed in~\cite{MCKN,MCAPPA2023,MCEntropy}.

A formal basis of the construction from~\cite{MCAPPA2021,MCKN,MCAPPA2023,MCEntropy} is given by the~following: 

\begin{lemma}
 {Consider four joint} %MDPI: %MDPI: We changed Lemma format as required. Please confirm whether the lemma content are all included in the Lemma environment.
 probabilities $p_{0_10_2}$, $p_{1_11_2}$, $p_{0_11_2}$, $p_{1_10_2}$, satisfying 
\be
\sum_{ab}p_{ab} &=& 1,\label{L2a}\\
\sum_{a}p_{aa_2} &=& \sum_{a}p_{a_1 a}=\frac{1}{2}.\label{L2b}
\ee
A sufficient condition for
\be
\sum_{ab}G(p_{ab}) &=& 1,\label{L2G}
\ee
is given by $G(p)=\frac{1}{2}g(2p)$, where $g$ satisfies Lemma~1. Any such $G$ has a fixed point at $p=1/4$.
\end{lemma}

A disadvantage of the construction based on Lemma~4 is its restriction to ``rotationally symmetric'' probabilities, i.e.,~those fulfilling (\ref{L2b}). Moreover, being in itself sufficient as a counterexample to Bell's theorem, it lacks the generality typical of arbitrary $k,l\in\mathbb{Z}$. 

The fundamental structure of the quantum probability model seems to be best described by Formula (\ref{P(a,b)}).

So far, the~angles occurring in singlet-state probabilities were interpretable as experimental parameters (angles between polarizers or Stern--Gerlach devices). But~what about arbitrary quantum states, even those described by infinite-dimensional Hilbert spaces? It turns out that the parameter in question can be interpreted in geometric terms, independently of the physical nature of the~problem.

\section{Fubini--Study Geodesic Distance as a Hidden~Variable}
\label{Fubini}

The scalar product $\langle a|b\rangle$ of two vectors belonging to some Hilbert space defines their Fubini--Study geodesic distance $\theta(a,b)$ \cite{Fubini,Study,Lane,Brody,Karol,Darek},
\be
|\langle a|b\rangle|^2=\langle a|a\rangle\langle b|b\rangle \cos^2 \theta(a,b)
.
\ee
Let $P_b$ be a projector, $|b\rangle=P_b|a\rangle$, and~$\langle a|a\rangle=1$, so that $\langle b|b\rangle=\langle a|b\rangle=\langle a|P_b|a\rangle=P(b|a)$ is a conditional quantum probability. The~geodesic distance between $|a\rangle$ and $|b\rangle$ satisfies
\be
|\langle a|b\rangle|^2=\langle a|P_b|a\rangle^2=\langle a|P_b|a\rangle \cos^2 \theta(a,b),
\ee
and thus,
\be
P(b|a)=\cos^2 \theta(a,b).
\label{123}
\ee
The formal angle $\theta(a,b)$ between the two vectors in the Hilbert space acquires a direct physical interpretation if $a$ and $b$ represent linear polarizations of photons: $\theta(a,b)$ becomes the angle between two polarizers. In~the analogous case of the electrons, $\theta(a,b)$ would represent one half of the angle between two Stern--Gerlach~devices. 

Next, let us rewrite (\ref{123}) as
\be
P(b|a)
=
\cos^2 \theta(a,b)=\sin^2\frac{\pi}{2}p(b|a)
=
g\big(p(b|a)\big)
=
\cos^2 \frac{\pi}{2}\big(1-p(b|a)\big),
\ee
where $g:[0,1]\to[0,1]$ is the bijection we have introduced in the context of the singlet state. 
Probabilities $p(b|a)$ and $P(b|a)=g(p(b|a))$ represent, respectively, the~hidden and the quantum neighboring levels of the hierarchy of (conditional) probabilities. The~hidden probability is thus directly related to the Fubini--Study geodesic distance,
\be
\theta(a,b) &=&\frac{\pi}{2}\big(1-p(b|a)\big),\\
p(b|a) &=&1-\frac{\theta(a,b)}{\pi/2},\\
q(b|a) &=&1-p(b|a)=\frac{\theta(a,b)}{\pi/2},
\ee
where $q(b|a)$ is the probability that two randomly chosen and intersecting straight lines intersect at an angle not exceeding $\theta(a,b)\in[0,\pi/2]$. 

{\it The Fubini--Study geodesic distance has been turned into  a classical measure of a subset of a quarter-circle\/}. It defines the whole hierarchy of probabilities, $g^k\big(p(b|a)\big)$, where $k=1$ is the quantum one. Note that $g(p)=\sin^2\frac{\pi}{2}p$ has been elevated to the role of a universal bijection, defining an arithmetic applicable to all the possible (pure) quantum states. Explicitly, we~find
\be
&\vdots&\nonumber\\
g^{-1}\big(p(b|a)\big) &=&\frac{1}{\pi/2}\arcsin\sqrt{1-\frac{\arccos\sqrt{P(b|a)}}{\pi/2}},\label{130}\\
g^0\big(p(b|a)\big) &=&1-\frac{\arccos\sqrt{P(b|a)}}{\pi/2},\\
g^1\big(p(b|a)\big) &=&\sin^2\frac{\pi}{2}\left(1-\frac{\arccos\sqrt{P(b|a)}}{\pi/2}\right)=P(b|a),\\
g^2\big(p(b|a)\big) &=&\sin^2\left(\frac{\pi}{2}P(b|a)\right),\\
g^3\big(p(b|a)\big) &=&\sin^2\left[\frac{\pi}{2}\sin^2\left(\frac{\pi}{2}P(b|a)\right)\right],\label{134}\\
&\vdots&\nonumber
\ee
Since $\langle a|P_b|a\rangle=P(b|a)$ is real, it can be written as a real quadratic form,
 \be
\langle a|P_b|a\rangle
=
\sum_{rs}\Re(a_r)A_{rs}\Re(a_s)+\sum_{rs}\Im(a_r)B_{rs}\Im(a_s)+\sum_{rs}\Re(a_r)C_{rs}\Im(a_s).\label{135}
\ee
Hence,
{\small \be
g^2\big(p(b|a)\big)
&=&
g^1\big(P(b|a)\big)\\
&=&
g^1\left(\sum_{rs}\Re(a_r)A_{rs}\Re(a_s)+\sum_{rs}\Im(a_r)B_{rs}\Im(a_s)+\sum_{rs}\Re(a_r)C_{rs}\Im(a_s)\right)\\
&=&
\bigoplus_{rs}
g\big(\Re(a_r)\big)\odot g\big(A_{rs}\big)\odot g\big(\Re(a_s)\big)
%\nonumber\\
%&\pp=&
\bigoplus_{rs}
 g\big(\Im(a_r)\big)\odot g\big(B_{rs}\big)\odot g\big(\Im(a_s)\big)
\nonumber\\
&\pp=&
\pp{\bigoplus_{rs}}
\bigoplus_{rs}
g\big(\Re(a_r)\big)\odot g\big(C_{rs}\big)\odot g\big(\Im(a_s)\big)
\label{138}
\\
&=&
\langle g(a)|\odot g(P_b)\odot|g(a)\rangle
=
\langle a_1|\odot_1 P_{b,1}\odot_1|a_1\rangle,
\label{139}
\ee}
where $\langle g(a)|\odot g(P_b)\odot|g(a)\rangle$ in (\ref{139}) is defined in a way that parallels the form of 
\be 
\langle a|P_b|a\rangle=\langle a_0|\odot_0P_{b,0}\odot_0|a_0\rangle
\ee
in (\ref{135}), but~with all the ``standard'' sums $+=\oplus_0$ and products $\cdot=\odot_0$ replaced by $\oplus_1$ and $\odot_1$, and~all the coefficients transformed by $g$.
In effect, the~difference between (\ref{135}) and (\ref{139}) is purely notational, as~one can write the whole hierarchy of probabilities in a ``quantum'' form as well,
\be
&\vdots&\nonumber\\
g^0\big(p(b|a)\big)
&=&
\langle a_{-1}|\odot_{-1} P_{b,-1}\odot_{-1}|a_{-1}\rangle,\label{134,,,}\\
g^1\big(p(b|a)\big)
&=&
\langle a_0|\odot_0 P_{b,0}\odot_0|a_0\rangle,\\
g^2\big(p(b|a)\big)
&=&
\langle a_1|\odot_1 P_{b,1}\odot_1|a_1\rangle,\\
g^3\big(p(b|a)\big)
&=&
\langle a_2|\odot_2 P_{b,2}\odot_2|a_2\rangle
\label{137,,,}
\\
&\vdots&\nonumber
\ee
This is the Copernican principle in action. The~choice of the ``quantum'' level of the hierarchy is just a matter of convention. In~fact, any formula from (\ref{130})--(\ref{134}) can represent quantum mechanics known from~textbooks.

It is perhaps more striking that any of these levels can be regarded as a hidden-variable level, where the hidden variable is given by an appropriate geodesic~distance. 

The concrete example of $g(p)=\sin^2\frac{\pi}{2}p$ can help us to understand the  structure of the whole hierarchy. We will see that, in~spite of the infinite dimension of the hierarchy, one effectively deals with a finite dimensional~structure.

\section{Effective Trunction of the Infinite Hierarchy of~Probabilities}
\label{SecFinal}

Figure~\ref{FigLast} explains why in spite of the infinite number of levels, those that statistically differ between one another may be limited to a finite ``band'' in the hierarchy. What it practically means is that if our level of the hierarchy is given by some $l$ (say, $l=0$) then, depending on the available precision of our experiments, we may restrict the analysis to a finite collection of probabilities. In~the example depicted in Figure~\ref{FigLast}, we can restrict the analysis to 31 levels,
\be
\{g^{-15}(p),\dots,g^{-1}(p),p,g(p),\dots,g^{15}(p)\},
\ee 
because the full infinite hierarchy is indistinguishable from
\be
\{\dots,g^{-15}(p),\dots,g^{-15}(p),\dots,g^{-1}(p),p,g(p),\dots,g^{15}(p),\dots,g^{15}(p),\dots\},
\ee 
When increasing $k$ in $g^k$, we effectively obtain a  theory that may look discrete, because~$g^{k}(p)$, $k>k_{\rm max}$, are indistinguishable from the red step function in Figure~\ref{FigLast}. For~$g^{k}(p)$, $k<k_{\rm min}$, we obtain an analogous behavior of the inverse~functions. 

Let us stress that the above argument for indistinguishability has been formulated only for probabilities,  $p\in[0,1]$, hence for $g(p)$, and~not for $g_\mathbb{R}(x)$, $x\not\in[0,1]$. In~principle, for~$x\not\in[0,1]$, all the levels of the hierarchy may be~distinguishable.

Notice that for this concrete $g(p)=\sin^2\frac{\pi}{2}p$, one finds $g^{15}(p)\approx 0$ if $p<1/2$, $g^{15}(1/2)=1/2$, and~$g^{15}(p)\approx 1$ if $p>1/2$. Thus, the~higher-level probabilities possess several obvious analogies to neural activation functions~\cite{activ}, making links between the hierarchical structure and the measurement problem even more intriguing. 
An observer who measures $g^{15}(p)$ probabilities ignores practically  all the events whose probability is smaller than 1/2, and~treats all $p>1/2$ as~certain. 

This type of behavior is the essence of learning algorithms. An~intriguing possibility occurs that $g(p)$ is a probability related to the act of learning that events with probability $p$ are true. Hence, the~natural question: Is the stabilization of large $k>0$ iterates on effectively the step function a formal counterpart of stabilization of self-observation, a~creation of self-awareness? 

For the negative iterates, instead of a threshold function we tend toward a ``white noise'': $g^{-15}(0)= 0$, $g^{-15}(1/2)= 1/2$, $g^{-15}(1)= 1$, and~$g^{-15}(p)\approx 1/2$, for~$0<p<1$.  The~lower levels of the hierarchy become less and less diverse from the point of view of a higher-level observer. Here, the~analogy is with observations of micro-scale events is quite evident. The~relativity of  probability becomes analogous to the ``relativity of smallness''---what is small to us, may be large for a bacteria or an~atom.

It is worth recalling that $g^{-15}(p)$ and $g^{15}(p)$ only {\it look\/} discrete due to our limited resolution---in reality, both maps are continuous bijections of $[0,1]$ into~itself. 

Now, what about experiment and laws of large numbers? Can they somehow discriminate between all these probabilities?

\begin{figure}

\includegraphics[width=10 cm]{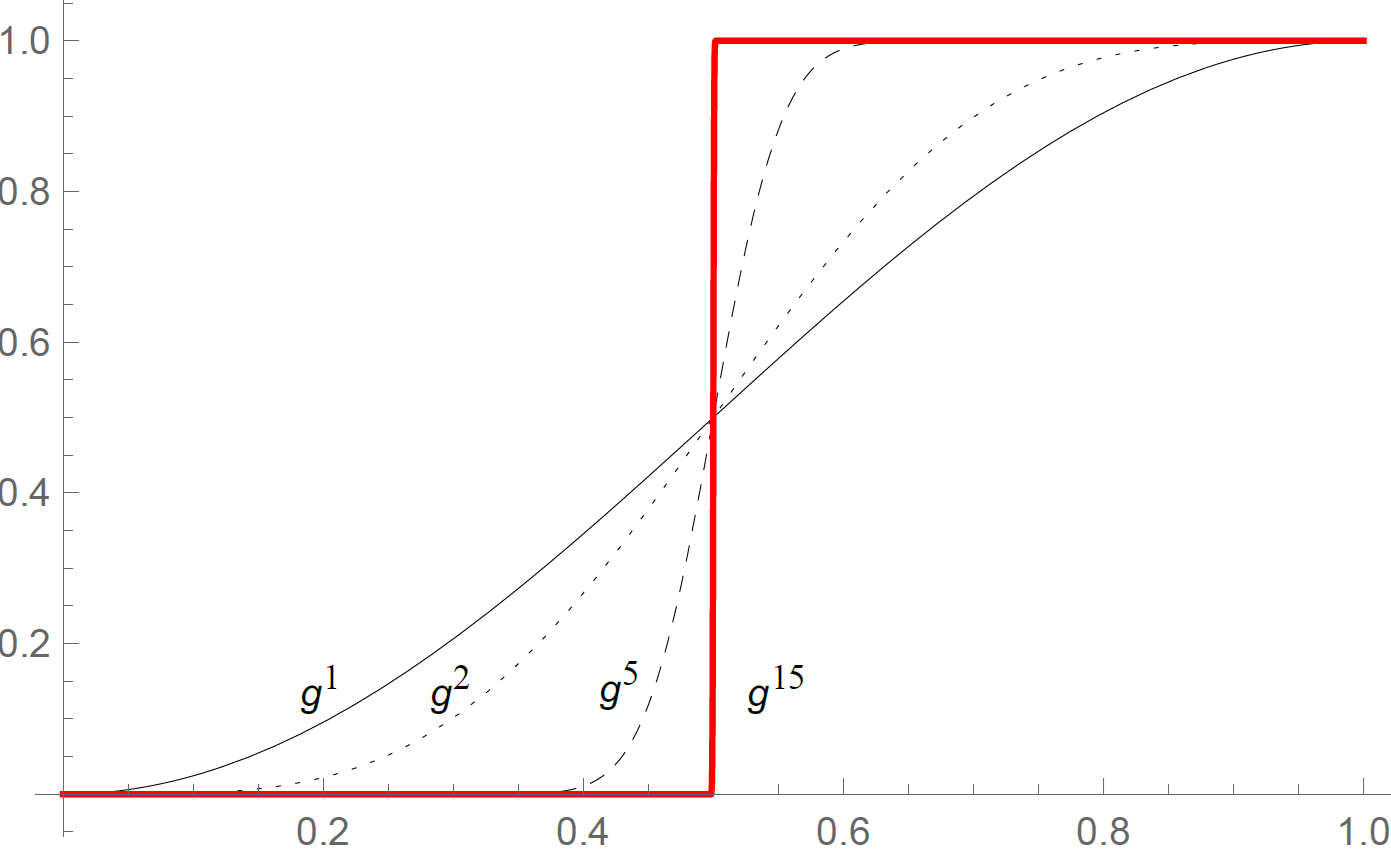}\\
\includegraphics[width=10 cm]{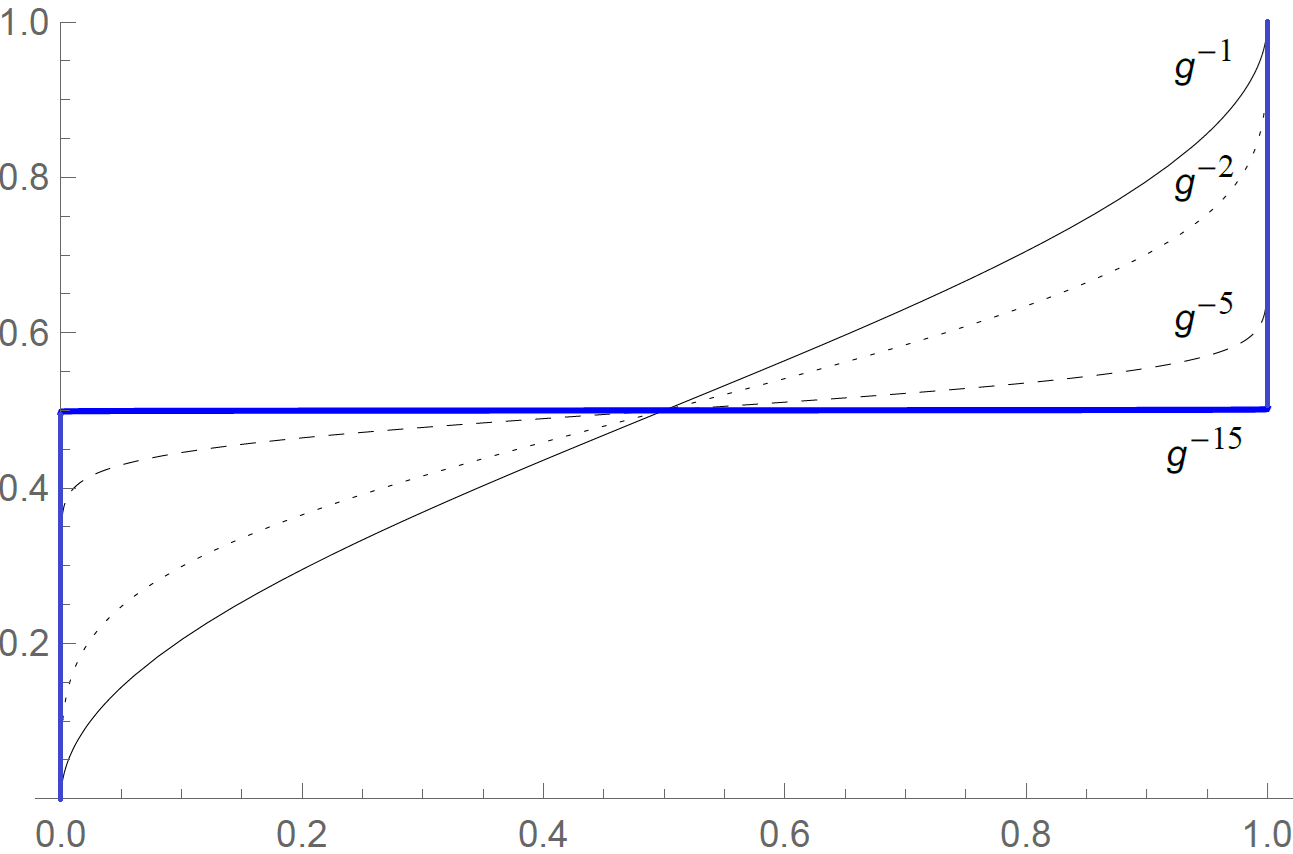}
\caption{A total of 1, 2, 5 and 15 iterations of $g(p)=\sin^2\frac{\pi}{2}p$ (upper plots). All the curves cross at $p=1/2$ and are of the sigmoidal form, analogously to activation functions occurring in learning algorithms. Is it just a coincidence, or~are there deeper connections to the problem of measurement, learning, or~consciousness? Iterates $g^k$ with $k>15$ are practically indistinguishable within the precision of the plot: They all look like the red step function. An~analogous phenomenon occurs for the negative iterates: $k=-1,-2,-5,-15$, but~here almost all events described by $g^{-15}(p)$ are equally probable, hence indistinguishable for level-0 observers (lower plots).  Effectively, even though the number of levels is infinite, the~distinguishable ones are restricted to a finite ``band'' $k_{\rm min}\le k\le k_{\rm max}$. Of~course, the~Copernican aspect of the hierarchy means that the same happens in a neighborhood of any $l$, and~not only $l=0$ depicted~here.}
\label{FigLast}
\end{figure}
\unskip   

\section{Hierarchical Laws of Large~Numbers}
\label{Section11}

Laws of large numbers formalize the relations between probabilities (real numbers), (natural) numbers of trials and successes, and~(rational) numbers of  their relative frequencies. However, as~we already know, all these notions are arithmetic dependent: a natural number $n_k=g^k(n)\in\mathbb{R}_k$ may not be a natural number from the point of view of some other $\mathbb{R}_l$, a~rational number $(n/m)_k=g^k(n/m)\in\mathbb{R}_k$ may not be a rational number from the point of view of $\mathbb{R}_l$, and~so on. The~most general law of large numbers should involve all the levels of the hierarchy simultaneously. Dealing with binary events, we need an appropriate generalization of  the Bernoulli law of large~numbers. 

To begin with, let us imagine we ``live'' in a world where all the possible computations are performed in terms of the arithmetic $\mathbb{R}_l$. If~we toss a coin, say, one hundred times, and~observe heads forty times, the~arithmetic formulation of the experiment involves $n_l=40_l$ heads in $N_l=100_l$ trials. The~experimental ratio is $n_l\oslash_l N_l=40_l\oslash_l100_l$. This is a rational number in $\mathbb{R}_l$.

{If the same experiment is described by an observer who employs arithmetic \mbox{$\mathbb{R}_j$, $j\neq l$}}, the~experimental ratio is given by $n_j\oslash_j N_j=40_j\oslash_j100_j$. In~terms of $g^l$ and $g^j$ we can write $40_l\oslash_l100_l=g^l(40/100)$ and $40_j\oslash_j100_j=g^j(40/100)$. 
Yet, if~we demanded \mbox{$g^l(40/100)=g^j(40/100)$}, it would imply that $g^{l-j}(40/100)=40/100$, i.e.,~$40/100$ is a fixed point of $g^{l-j}$. Since the same argument can be applied to any rational number, one arrives at the conclusion that the trivial case $g(x)=g^0(x)$ is the only~solution. 

One concludes that a nontrivial $g$  generically implies $n_j\oslash_j N_j\neq n_l\oslash_l N_l$ for $l\neq j$. In~other words, the~same experiment can be described by different probabilities, \mbox{$p_l=g^l(p)\neq p_j=g^j(p)$}, although~from the frequentist perspective both descriptions involve forty successes in one hundred trials. We inevitably arrive at  the whole~hierarchy. 

This is my tentative interpretation of the hierarchical structure.
However, the~links with neural activation functions deserve a separate~study.

In order to formulate a generalized  Bernoulli law of large numbers, we 
have to estimate the probability that 
\be
\left|g^k(p)\ominus_l n_l \oslash_l N_l\right|
=
\left|g^l\left(g^{k-l}(p)-n/N\right)\right|
=
g^l\left(\left|g^{k-l}(p)-n/N\right|\right)
\ge
\varepsilon_l=g^l(\varepsilon).
\label{116''''}
\ee
The modulus is defined in $\mathbb{R}_l$ in the standard way,
\be
|x|=\left\{
\begin{array}{cl}
x & \textrm{if $x\ge 0_l$}\\
\ominus_l x & \textrm{if $x< 0_l$}
\end{array}
\right.
,
\ee
where we keep in mind that, by~assumption, $0_l=g^l(0)=0$ and the ordering relation is unaffected by a strictly increasing $g$.
Inequality (\ref{116''''}) effectively boils down to 
\be
\left|g^{k-l}(p)-n/N\right|
\ge
\varepsilon.
\label{117''''}
\ee
Next, we 
note that probabilities depicted in the lower part of  Figure~\ref{Fig2} are normalized in consequence of the identity
\be
\left(
g^k(p)\oplus_l g^k(q)
\right)^{N_l}
&=&
\bigoplus_{n=0}^N{}\!_l
{{N_l}\choose{n_l}}
\odot_l 
g^k(q)^{(N-n)_l}\odot_l g^k(p)^{n_l}
=1_l=1,\\
{{N_l}\choose{n_l}}
&=&
g^l\left[{N}\choose{n}\right].
\ee
The probability 
\be
p(n_l,N_l)_k
&=&
{{N_l}\choose{n_l}}
\odot_l 
g^k(q)^{(N-n)_l}\odot_l g^k(p)^{n_l}
\\
&=&
g^l\left[
{{N}\choose{n}}
g^{k-l}(q)^{N-n}g^{k-l}(p)^{n}
\right]
\ee
corresponds to  $n_l$ sucessess in $N_l$ trials.
The expected number of successes and the corresponding variance read,
\be
\langle n_l\rangle_k
&=&
\bigoplus_{n=0}^N{}\!_l
n_l \odot_l {{N_l}\choose{n_l}}
\odot_l 
g^k(q)^{(N-n)_l}\odot_l g^k(p)^{n_l}
\\
&=&
g^l\left[Ng^{k-l}(p)\right]
=
N_l\odot_l g^{k}(p),\\
\langle n_l^{2_l}\rangle_k\ominus_l
\langle n_l\rangle_k^{2_l}
&=&
g^l\left(Ng^{k-l}(p)g^{k-l}(q)\right)
=
N_l\odot_l 
g^{k}(p)\odot_l g^{k}(q)
\label{86''_}\\
&=&
N_l^{2_l}
 \odot_l
\bigoplus_{n=0}^N{}\!_l
\left[
n_l\oslash_l N_l
\ominus_l
g^k(p)
\right]^{2_l}
\odot_l
p(n_l,N_l)_k
\label{87''_}.
\ee
Applying $g^{-l}$ to (\ref{86''_}) and (\ref{87''_}), we find
%%%%%%%%%%%%%%%%%%%%%%%
\be
g^{k-l}(p)g^{k-l}(q)/N
&=&
\sum_{n=0}^N
\left[
n/N
-
g^{k-l}(p)
\right]^{2}
{{N}\choose{n}}
g^{k-l}(q)^{N-n}g^{k-l}(p)^{n}
.
\ee
Now, let $n\in \mathbb{N}_{\varepsilon,g^{k-l}(p),N}$ if 
$
\left|n/N
-
g^{k-l}(p)
\right|
\ge
\varepsilon$. 
Then,
\be
g^{k-l}(p)g^{k-l}(q)/N
&\ge&
\sum_{n\in \mathbb{N}_{\varepsilon,g^{k-l}(p),N}}
\left[
n/N
-
g^{k-l}(p)
\right]^{2}
{{N}\choose{n}}
g^{k-l}(q)^{N-n}g^{k-l}(p)^{n}
\\
&\ge&
\varepsilon^2
\sum_{n\in \mathbb{N}_{\varepsilon,g^{k-l}(p),N}}
{{N}\choose{n}}
g^{k-l}(q)^{N-n}g^{k-l}(p)^{n}
\\
&=&
\varepsilon^2
p
\left(
n\in \mathbb{N}_{\varepsilon,g^{k-l}(p),N}
\right),
\ee
where $p
\left(
n\in \mathbb{N}_{\varepsilon,g^{k-l}(p),N}
\right)
\in\mathbb{R}_0$ is the 0th-level probability that 
$\big|n/N-g^{k-l}(p)\big|
\ge
\varepsilon$. In~this way we have arrived at the standard Bernoulli law of large numbers in $\mathbb{R}_0$,
\be
p
\left(
n\in \mathbb{N}_{\varepsilon,g^{k-l}(p),N}
\right)\le
\frac{g^{k-l}(p)g^{k-l}(q)}{N\varepsilon^2}.\label{137''''}
\ee
Of course, the~left-hand side of (\ref{137''''}) cannot be greater than 1, so the number of trials $N$ must be chosen so that 
\be
\frac{g^{k-l}(p)g^{k-l}(q)}{\varepsilon^2}\le N.
\ee
For $p_l=g^l(p)$ we find, denoting $\varepsilon_l=g^l(\varepsilon)$, $N_l=g^l(N)$,
\be
p_l
\left(
n\in \mathbb{N}_{\varepsilon,g^{k-l}(p),N}
\right)
&\le&
g^l\left(\frac{g^{k-l}(p)g^{k-l}(q)}{N\varepsilon^2}\right)
=
g^l\left(\frac{g^{k-l}(p)g^{k-l}(q)}{g^{-l}(N_l)g^{-l}(\varepsilon_l)^2}\right)
\nonumber
\\
&=&
g^k(p)\odot_l g^k(q)\oslash_l\left(N_l\odot_l\varepsilon_l^{2_l}\right),
\ee
for any $k\in\mathbb{Z}$.

\begin{figure}

\includegraphics[width=11 cm]{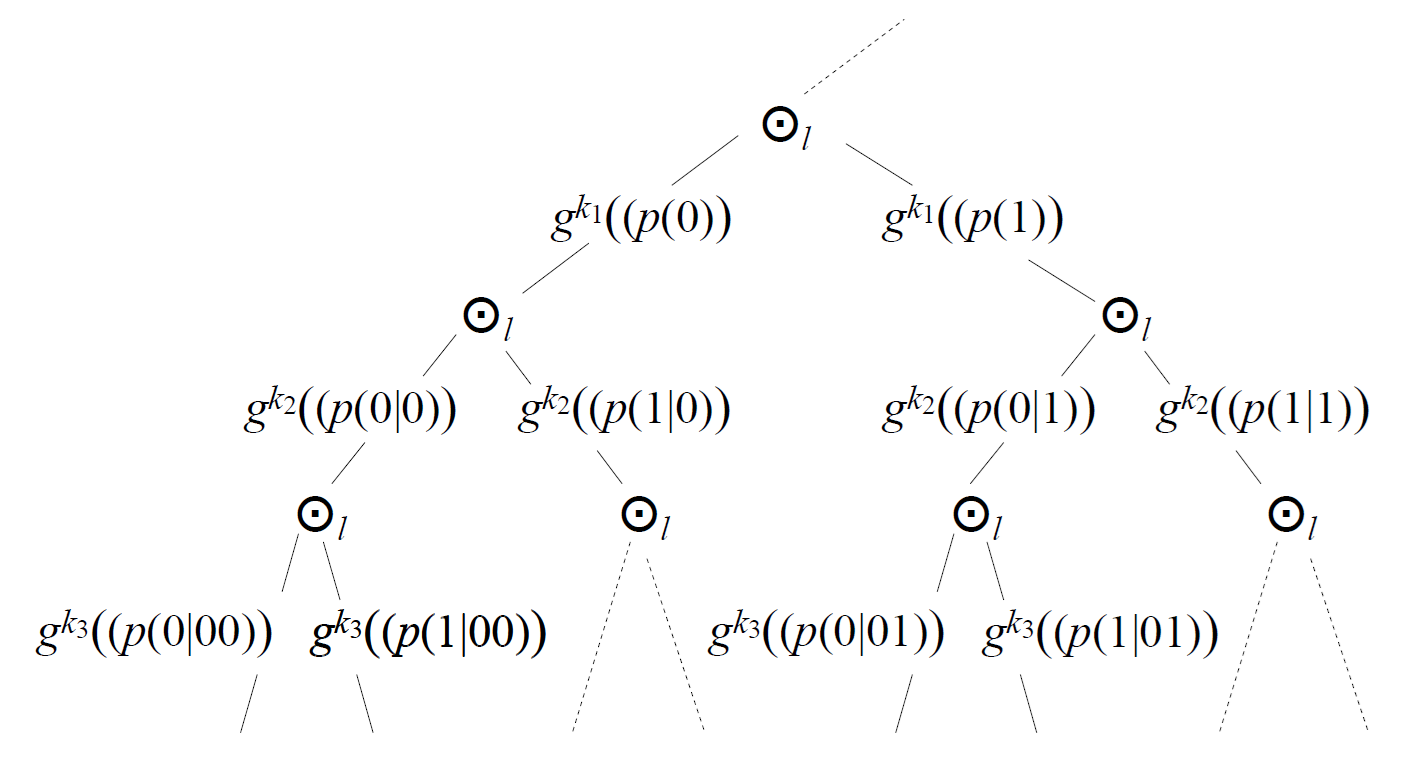}\\
\includegraphics[width=10 cm]{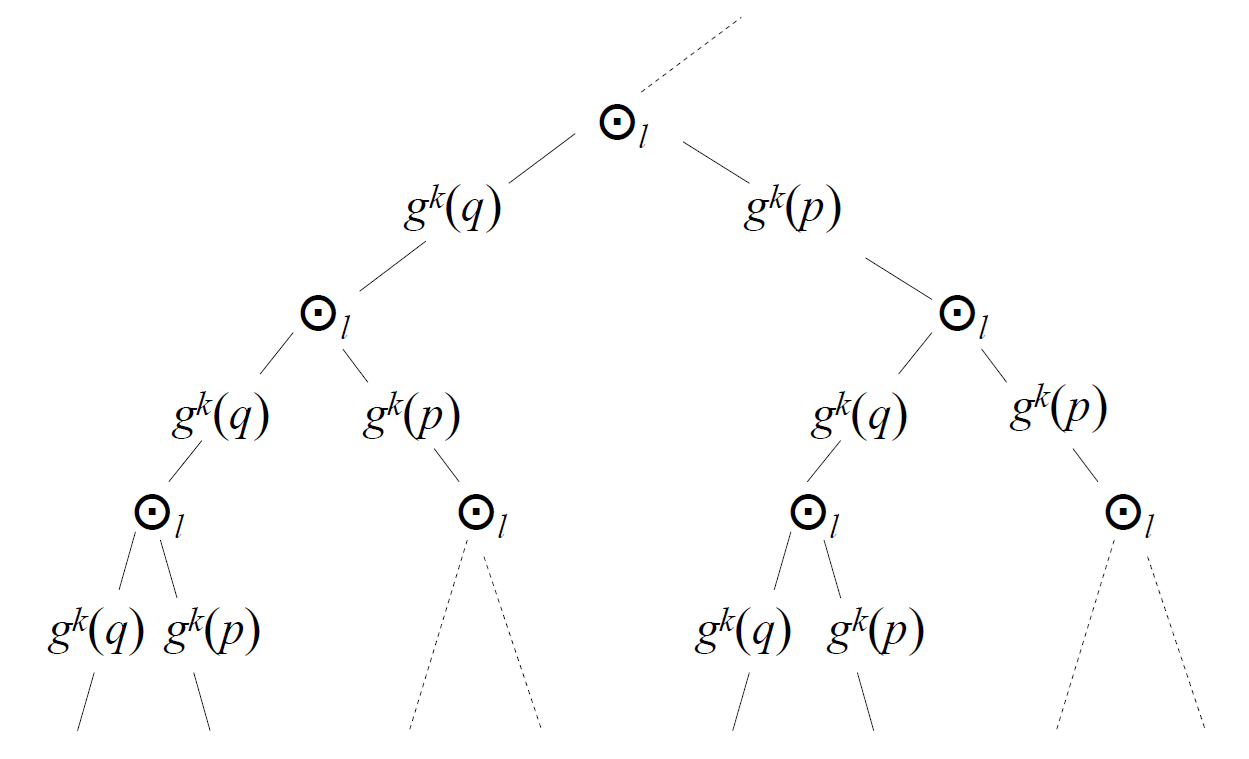}
\caption{The upper diagram: An $\mathbb{R}_l$-valued branch of a binary tree of conditional probabilities. This is how one can include events with more results than just two.  Assuming independent events and the same value of all $k_j$ (the lower diagram), we can derive a hierarchical analog  of  the Bernoulli law of large numbers. Laws of large numbers are the places where theory and experiment~meet.}
\label{Fig2}
\end{figure}

In order to have a feel of the influence of $l\in\mathbb{Z}$ on the rate of convergence of experimental ratios to probabilities, consider the simple case of a symmetric coin,  $p=q=1/2$, and~the universal quantum bijection $g(x)=\sin^2\frac{\pi}{2}x$. Since $g^{k-l}(1/2)=1/2$ for any $k,l$, we have to estimate %Please check intended meaning has been retained.

\be
p_l
\left(
n\in \mathbb{N}_{\varepsilon,g^{k-l}(p),N}
\right)
&\le&
g^l\left(\frac{1}{4N\varepsilon^2}\right), 
\label{139''''}\\
\frac{1}{4\varepsilon^2}&\le& N.
\ee
Figure~\ref{Fig3} illustrates the right-hand side of (\ref{139''''}) for $\varepsilon=0.1$ and $25\le N\le 75$, for~the first four iterates of $g$, from~$g^1(x)=\sin^2\frac{\pi}{2}x$ to 
\be
g^4(x)=\sin^2\frac{\pi}{2}\left(\sin^2\frac{\pi}{2}\left(\sin^2\frac{\pi}{2}\left(\sin^2\frac{\pi}{2}x\right)\right)\right).
\ee
The graphs are intriguing. Their interpretation is additionally obscured by the fact that Wolfram Mathematica operates in the arithmetic $\mathbb{R}_0$, which is not used by any of the four observers.  The~problem requires further~studies.

\begin{figure}

\includegraphics[width=12 cm]{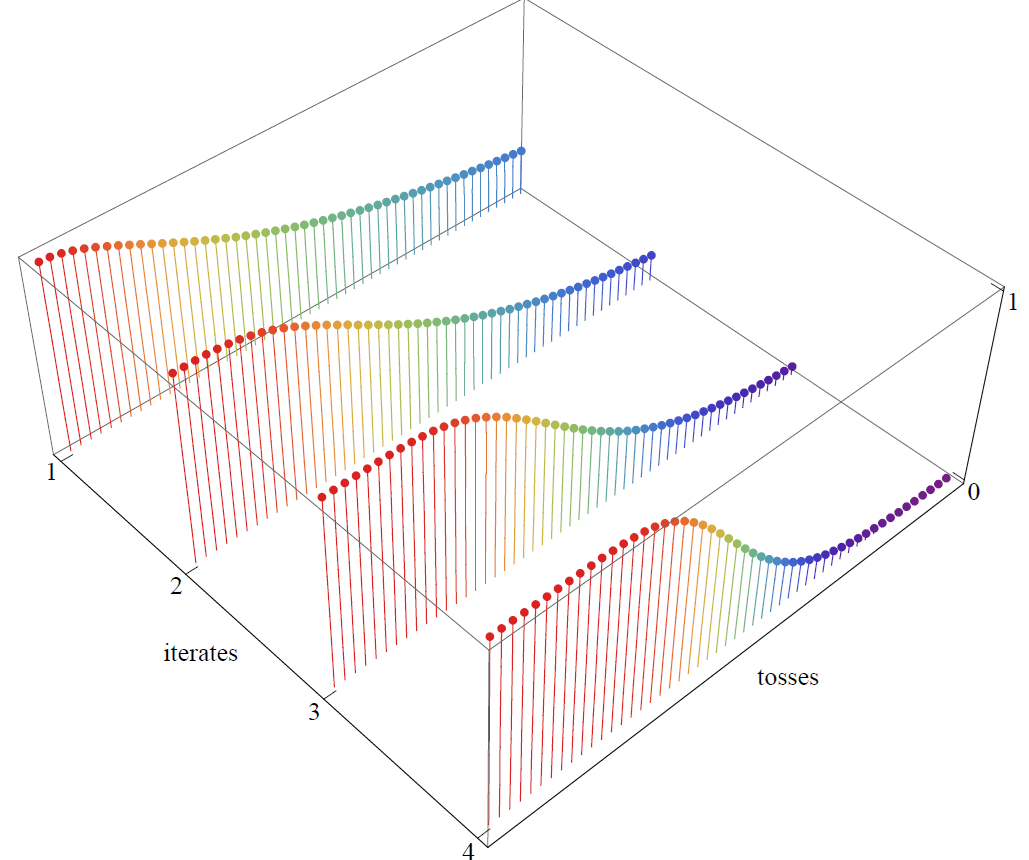}
\caption{{Hierarchical} %MDPI: Please confirm whether an explanation of the colors needs to be added to the figure caption.
 law of large numbers in action. Upper bound on probability of disagreement between theory and experiment in $N$ tosses of a symmetric coin for four different arithmetics $\mathbb{R}_l$ of the observer. Plot of the  right-hand side of (\ref{139''''}) with $\varepsilon=0.1$, for~the four iterates $g^l$, $l=1,2,3,4$, of~$g(x)=\sin^2\frac{\pi}{2}x$. The~number of coin tosses $25\le N\le 75$. Plots are made in the arithmetic $\mathbb{R}_0$, implicitly assumed in {Wolfram Mathematica 14.} %MDPI: Please state which version of the software was used.
}
\label{Fig3}
\end{figure}

\section{Hierarchical Approach to Bell's Theorem---Revisited}
\label{Section12}

If we are able to reconstruct singlet-state probabilities in a hidden-variable way, it means that Bell's inequality (in any form) cannot be proved for the model. In~the hierarchical context the obstacle for proving the inequality lies in the lack of the $k$-level additivity of the $l$-level integrals, if~$l\neq k$. The~usual derivation, when seen from the hierarchical perspective, assumes $\oplus_0$-additivity of ${\rm D}_1\lambda$ integrals, which is untrue for a nontrivial $g$, and~$g(p)=\sin^2\frac{\pi}{2}$ in particular, hence the inequality derived at level zero does not apply to level 1: Level-0 formulas are ``violated'' by level-1 probabilities (and the other way around). 

Let us see how it works. Consider the joint probabilities
\be
P(a_1,a_2)
&=&
P(\textrm{first $a_1$ then $a_2$})
=
g\big(p(a_2|a_1)\big)g\big(p(a_1)\big)
=
g\big(p(a_1|a_2)\big)g\big(p(a_2)\big)
\nonumber\\
&=&
P(\textrm{first $a_2$ then $a_1$})=P(a_2,a_1),\label{152}
\ee
where we assume the independence of the order in which the measurements are performed. This is typical of the scenarios involving ``observer 1 measuring $a_1$'' (``Alice'') and  
 ``observer 2 measuring $a_2$'' (``Bob'') who are space-like separated and thus the order is~undefined. 

Now, we will  derive an analog of the Clauser--Horne inequality~\cite{CH}. We will work with probabilities (\ref{152}). Let us stress that an analogous derivation was presented in~\cite{MCKN}, but~was based on the form occurring in (\ref{94''}), that is by means of the bijection $G$. The~derivation we will discuss now is based on $g(x)$, and~not on  $G(x)=\frac{1}{2}g(2x)$. 
Why? Because we want a proof that is easy to generalize to any $k,l\in\mathbb{Z}$.

We assume a local-hidden variable form of the probabilities that occur at the hidden level (level zero), hence
\be
p(a_2|a_1)
&=&
\frac{\int\chi_{a_1,00}(\lambda)\chi_{a_2,00}(\lambda)\rho_{00}(\lambda)\,{\rm d}\lambda}
{\int\chi_{a_1,00}(\lambda)\rho_{00}(\lambda)\,{\rm d}\lambda }
,\\
p(a_1)
&=&
\int\chi_{a_1,00}(\lambda)\rho_{00}(\lambda)\,{\rm d}\lambda.
\ee
Level-one conditional probabilities 
\be
g\big(p(a_2|a_1)\big)
&=&
g\left(
\frac{\int\chi_{a_1,00}(\lambda)\chi_{a_2,00}(\lambda)\rho_{00}(\lambda)\,{\rm d}\lambda}
{\int\chi_{a_1,00}(\lambda)\rho_{00}(\lambda)\,{\rm d}\lambda}
\right)
\\
&=&
\int\chi_{a_1,11}\odot_1\chi_{a_2,11}\odot_1\rho_{11}(\lambda)\,{\rm D}_1\lambda
\oslash_1
\int\chi_{a_1,11}\odot_1\rho_{11}(\lambda)\,{\rm D}_1\lambda,
\ee
can be rewritten in several useful forms. First of all, introducing the reduced (conditional) probability density we obtain the ``projection postulate'',
\be
\rho_{11}(\lambda)
&\mapsto&
\rho_{a_1,11}(\lambda)
=
\chi_{a_1,11}\odot_1\rho_{11}(\lambda)
\oslash_1
\int\chi_{a_1,11}\odot_1\rho_{11}(\lambda'){\rm D}_1\lambda',
\label{159}\\
g\big(p(a_2|a_1)\big)
&=&
\int\chi_{a_2,11}\odot_1\rho_{a_1,11}(\lambda)\,{\rm D}_1\lambda.
\ee
Secondly, we can explicitly express the conditional probability in a local Clauser--Horne form (in the arithmetic $\mathbb{R}_1$),
\be
g\big(p(a_2|a_1)\big)
&=&
\int x_{a_1,11}(\lambda)\odot_1 y_{a_2,11}(\lambda)\odot_1\rho_{11}(\lambda)\,{\rm D}_1\lambda,
\ee
where
\be
x_{a_1,11}(\lambda)
&=&
\chi_{a_1,11}(\lambda)\oslash_1 \int\chi_{a_1,11}(\lambda')\odot_1\rho_{11}(\lambda'){\rm D}_1\lambda'\\
&=&
g_\mathbb{R}\left(
\frac{\chi_{a_1,00}(\lambda)}{\int\chi_{a_1,00}(\lambda')\rho_{00}(\lambda'){\rm d}\lambda'}
\right),
\\
y_{a_2,11}(\lambda)
&=&
\chi_{a_2,11}(\lambda)=g\big(\chi_{a_2,00}(\lambda)\big)=\chi_{a_2,00}(\lambda),
\ee
(because $g(0)=0$, $g(1)=1$). 

Repeating step by step the derivation of the Clauser--Horne inequality~\cite{CH}, but~here in the arithmetic $\mathbb{R}_1$, we can derive an analogous inequality which {\it must\/} be satisfied at the quantum level of the hierarchy. Such an inequality {\it cannot be violated\/} by quantum probabilities.
For simplicity let us reduce the analysis to singlet-state probabilities and $g_\mathbb{R}(n)=n$ for any $n\in\mathbb{Z}$. Then,
\be
\int x_{a_1,11}(\lambda)\odot_1\rho_{11}(\lambda){\rm D}_1\lambda
&=&
1,
\\
\int y_{a_2,11}(\lambda)\odot_1\rho_{11}(\lambda){\rm D}_1\lambda
&=&
1/2,
\ee
\be
0\le x_{a_1,11}(\lambda)
\le
X=g_\mathbb{R}\left(
\frac{1}{\int\chi_{a_1,00}(\lambda')\rho_{00}(\lambda'){\rm d}\lambda'}
\right)=2,
\ee
\be
0\le y_{a_2,11}(\lambda)\le Y=1,
\ee
and
\be
g\big(p(a_2|a_1)\big)
=g\big(p(a_1|a_2)\big)
.
\ee
Next, we consider the Clauser--Horne linear combination 
{\footnotesize \be
CH(\lambda)
&=&
x_{a_1,11}\odot_1 y_{b_2,11}(\lambda)
\ominus_1
x_{a_1,11}\odot_1 y_{b'_2,11}(\lambda)
\oplus_1
x_{a'_1,11}\odot_1 y_{b_2,11}(\lambda)
\oplus_1
x_{a'_1,11}\odot_1 y_{b'_2,11}(\lambda)
\nonumber\\
&\pp=&
\ominus_1
x_{a'_1,11}\odot_1 Y
\ominus_1
X\odot_1 y_{b_2,11}(\lambda).
\ee}Repeating in $\mathbb{R}_1$ the reasoning from~\cite{CH}, we obtain
\be
-2\le CH(\lambda)\le 0.
\ee
$\mathbb{R}_1$-multiplying the latter by $\rho_{11}(\lambda)$, integrating with ${\rm D}_1\lambda$, and~taking into account the $\mathbb{R}_1$-linearity of the ${\rm D}_1\lambda$ integral, we find
\be
0\le 
g\big(p(a_1|b_2)\big)
\ominus_1
g\big(p(a_1|b'_2)\big)
\oplus_1
g\big(p(a'_1|b_2)\big)
\oplus_1
g\big(p(a'_1|b'_2)\big)
\le 2.\label{170}
\ee
Notice that inequality (\ref{170}) involves conditional probabilities, as~opposed to the original Clauser--Horne one which was based on joint probabilities. The~inequalities derived in the arithmetic induced by $G(x)$ and discussed in~\cite{MCAPPA2021,MCKN} were also based on joint probabilities. However, joint probabilities involve the ``macroscopic'' level-0 multiplication of $1/2$ by $\cos^2(\alpha/2)$, whereas the conditional probabilities involve only the arithmetic of the ``microscopic'' level-1 probability $\cos^2(\alpha/2)$.

When investigating the violation of inequalities such as (\ref{170}) one should keep in mind the difference between $g(x)=\sin^2\frac{\pi}{2}x$, for~$x\in[0,1]$, and~its extension $g_\mathbb{R}$(x) beyond the interval $[0,1]$. Here, (\ref{170}) is derived under the assumption that $g_\mathbb{R}(n)=n$, for~any integer $n$. Readers interested in explicit examples of $g_\mathbb{R}$ may consult~\cite{MCAPPA2021,MCAPPA2023,MCKN}.

The inequality that can indeed be violated is 
\be
0\le 
g\big(p(a_1|b_2)\big)
-
g\big(p(a_1|b'_2)\big)
+
g\big(p(a'_1|b_2)\big)
+
g\big(p(a'_1|b'_2)\big)
\le 2,
\label{172}
\ee
but it {\it {cannot be proved}%MDPI: Please confirm if the italics are necessary; if not, please remove them. The following highlights are the same.
\/} for the model, so is simply untrue. The technical difficulty in proving (\ref{172}) is the lack of $\mathbb{R}_0$-linearity of 
the ${\rm D}_1\lambda$ integral. 

The notion of ``violation'' of a formula is, in~my opinion, very confusing. In~the same sense one could say that the real-number inequality $x^2\ge 0$ is violated by complex numbers. Instead of saying that $i^2=-1$ violates $x^2\ge 0$ one rather says that $x^2\ge 0$ cannot be proved for all $x\in\mathbb{C}$. The~same happens with the Bell inequality, derived in $\mathbb{R}_0$ but not valid in $\mathbb{R}_1$.
On the other hand, the~inequalities that {\it can\/} be derived in $\mathbb{R}_1$ are never ``violated'' in $\mathbb{R}_1$, but~certainly will be untrue in some other $\mathbb{R}_k$.

\section{Interference, Propagators, Dynamics$\dots$}
\label{Section13}

Formulas (\ref{134,,,})--(\ref{137,,,}) show that the conditional probabilities can be written in scalar-product forms,
 \be
&\vdots&\nonumber\\
g^0\big(p(b|a)\big)
&=&
\langle b_{-1}|\odot_{-1}|a_{-1}\rangle=\langle b|a\rangle_{-1},\label{134,,,,}\\
g^1\big(p(b|a)\big)
&=&
\langle b_0|\odot_0|a_0\rangle=\langle b|a\rangle_{0},\\
g^2\big(p(b|a)\big)
&=&
\langle b_1|\odot_1|a_1\rangle=\langle b|a\rangle_{1},\\
g^3\big(p(b|a)\big)
&=&
\langle b_2|\odot_2|a_2\rangle=\langle b|a\rangle_{2},
\label{137,,,,}
\\
&\vdots&\nonumber
\ee
where we have introduced the compact notation,
\be
\langle b_k|\odot_k|a_k\rangle=\langle b|a\rangle_{k}=g\left(\langle b|a\rangle_{k-1}\right).
\ee
These concrete scalar products are real. However, a~complex scalar product can be always treated as a pair of reals with, in~principle, different arithmetics for real and imaginary parts (see Appendix \ref{appa}). This type of generalized complex numbers was applied to non-Newtonian Fourier analysis on fractals~\cite{ACK2016}, and~proved very useful in circumventing certain impossibility theorems about Fourier transforms on the triadic Cantor set. 
Scalar products (\ref{134,,,,})--(\ref{137,,,,}) when generalized to complex numbers (see Appendix \ref{appa}) can be used to generalize Feynman's path integral formalism to its hierarchical form, ultimately leading to propagators and time~evolution. 

We leave it for a future~paper.

\section{An Open~Ending}
\label{Section14}

Standard modern physics involves a three-level hierarchy: quantum, classical and cosmological.
As human observers, we are positioned at the center of this hierarchy, but~the connections with the remaining two levels remain unclear. We do not understand how is it that we observe quantum properties (the measurement problem). Similarly, we do not understand our relation with the large-scale universe (the dark energy problem). In~both cases the arithmetic freedom is probably essential~\cite{MCAPPA2021,MCFS2021} but generally overlooked by our scientific~community.

Bell's theorem is generally believed  to eliminate levels lower than the quantum one, but~the hierarchical picture questions this viewpoint: Quantum and classical probabilities typical of the singlet state belong to neighboring levels in the hierarchy---{\it {any}%MDPI: Please confirm if the italics are necessary; if not, please remove them. The following highlights are the same.
\/} two neighboring levels. Elimination of any of the levels, thus, would destroy the whole hierarchical structure, all quantum levels included~\cite{MCAPPA2023}.

To the best of my knowledge, the~first systematic study of generalized arithmetics in physics was initiated by my paper~\cite{MC}, in~which the relativity of arithmetic was interpreted in terms of a fundamental symmetry. 
However, I merely rediscovered a structure that had previously been introduced to calculus by Grossman and Katz 
(non-Newtonian calculus)~\cite{GK,G79,G79'}, Maslov (idempotent analysis) \cite{Maslov} and Pap (g-calculus) \cite{P}.  The~origins of the idea of generalized arithmetic and calculus can be traced back to the works of Volterra on the product integral~\cite{Volterra},   Kolmogorov~\cite{K1930}, and Nagumo~\cite{N1930} on generalized means, and~R\'enyi on generalized entropies~\cite{R}.  Studies of a nonstandard number theory were initiated by Rashevsky~\cite{Rashevsky} and, in~a concrete form of non-Diophantine arithmetic, developed by  Burgin~\cite{Burgin,Burgin',Burgin'',Burgin3}. Generalized forms of arithmetic can be found in Bennioff's attempts to formulate a coherent theory of physics and mathematics~\cite{B2002,B2002',B2002'',B2002''',B2002''''}.  Mathematical constructions such as Lad's impediment functions~\cite{Lad}, cepstral signal analysis~\cite{cep,cep1}, fractal $F^\alpha$-calculus~\cite{Gangal,Gangal',A4,A4',A6,Alireza}, or~nonextensive statistics~\cite{TMP,Kan,Nivanen,K1,K1'}, involve certain formal elements analogous to non-Newtonian integration or differentiation. 
The first application of non-Newtonian calculus to probability of which I am aware was provided by Meginniss in his analysis of the objectivity of $p$ versus the subjectivity of $g(p)$, with~applications to gambling theory~\cite{Meginniss}.  Another field in which generalized arithmetic and non-Newtonian calculus are starting to attract attention is mathematical finance~\cite{Carr1,Carr2}. From~my personal perspective, the~most important achievements of the new formalism include circumventing the limitations of Bell's theorem  and Tsirelson bounds  in quantum mechanics~\cite{MCAPPA2021,MCKN}; the arithmetic of time, which appears to eliminate dark energy from cosmology in the same way that the arithmetic of velocities eliminated the luminiferous aether from special relativity~\cite{MCFS2021}; formulating wave propagation along fractal coastlines~\cite{Czachor2019}; and overcoming the limitations of Fourier analysis on Cantor sets~\cite{ACK2016}.

The two most important observations of the present study seem to be the interpretation of the singlet-state probabilities in terms of several different arithmetic levels occurring in a single Formula (\ref{P(a,b)}),
\be
P(a,b)=\underbrace{g\big(\overbrace{p(a|b)}^{\textrm{hidden}}\big)}_{\textrm{quantum}}\underbrace{\odot_0}_{\textrm{macroscopic}} \underbrace{g\big(\overbrace{p(b)}^{\textrm{hidden}}\big)}_{\textrm{quantum}}.
\ee
and the possible links with neural network learning~algorithms. 

The hierarchical structure is clearly ``there''. What we have understood so far is just the tip of the iceberg.
%AAAAAAAAAAAAAAAAAAAAAAAAAAAAAAAAAAAAAAAAAAAAAA

\vspace{6pt}

\acknowledgments{Calculations in Mathematica were carried out at the Academic Computer Center in Gda{\'n}sk (CI TASK project~pt01234).}

\appendix
\section[\appendixname~\thesection]{}\label{appa}

\subsection[\appendixname~\thesubsection]{{Proof} %MDPI: 1. Appendix sholud be start with A. Please confirm this revision. 2. In order to cite Appendix in maintext, we changed following headings into heading level 2. Please confirm this revision.
 of (\ref{86''}):}

\vspace{-14pt}
\centering %% If there is a figure in wide page, please release command \centering, for Table, ``\textwidth" should be ``\fulllength"
\be
\frac{{\rm D}_l A_{lk}(x)}{{\rm D}_k x}
&=&
\lim_{\delta\to 0}
\Big(
A_{lk}(x\oplus_k\delta_k)\ominus_l A_{lk}(x)
\Big)
\oslash_l \delta_l
\nonumber\\
&=&
\lim_{\delta\to 0}
\Big(
g^{l-m}\circ  A_{mn}\circ f^{k-n}(x\oplus_k\delta_k)\ominus_l g^{l-m}\circ  A_{mn}\circ f^{k-n}(x)
\Big)
\oslash_l \delta_l
\nonumber\\
&=&
\lim_{\delta\to 0}
g^l\Big(
g^{-l+l-m}\circ  A_{mn}\circ f^{k-n}\circ g^k\big(f^k(x)+f^k(\delta_k)\big)- g^{-l+l-m}\circ  A_{mn}\circ f^{k-n}(x)
\Big)
\oslash_l \delta_l
\nonumber\\
&=&
\lim_{\delta\to 0}
g^l\Big(
g^{-m}\circ  A_{mn}\circ f^{-n}\big(f^k(x)+f^k(\delta_k)\big)- g^{-m}\circ  A_{mn}\circ f^{k-n}(x)
\Big)
\oslash_l \delta_l
\nonumber\\
&=&
\lim_{\delta\to 0}
g^l\Big[\Big(
g^{-m}\circ  A_{mn}\circ f^{-n}\big(f^k(x)+\delta\big)- g^{-m}\circ  A_{mn}\circ f^{k-n}(x)
\Big)
/\delta\Big]
\nonumber\\
&=&
\lim_{\delta\to 0}
g^{l-m}\circ g^m\Big[f^{m}\circ g^m\Big(
g^{-m}\circ  A_{mn}\circ f^{-n}\big(f^k(x)+\delta\big)- g^{-m}\circ  A_{mn}\circ f^{k-n}(x)
\Big)
/f^{m}(\delta_m)\Big]
\nonumber\\
&=&
\lim_{\delta\to 0}
g^{l-m}\Big[g^m\Big(
g^{-m}\circ  A_{mn}\circ f^{-n}\big(f^k(x)+\delta\big)- g^{-m}\circ  A_{mn}\circ f^{k-n}(x)
\Big)
\oslash_m\delta_m\Big]
\nonumber\\
&=&
g^{l-m}\lim_{\delta\to 0}
\Big[\Big(
A_{mn}\circ g^{n}\big(f^k(x)+\delta\big)\ominus_m A_{mn}\circ f^{k-n}(x)
\Big)
\oslash_m\delta_m\Big]
\nonumber\\
&=&
g^{l-m}\lim_{\delta\to 0}
\Big[\Big(
A_{mn}\circ g^{n}\big(f^n\circ f^{k-n}(x)+f^n(\delta_n)\big)\ominus_m A_{mn}\circ f^{k-n}(x)
\Big)
\oslash_m\delta_m\Big]
\nonumber\\
&=&
g^{l-m}\lim_{\delta\to 0}
\Big[\Big(
A_{mn}\big(f^{k-n}(x)\oplus_n \delta_n\big)\ominus_m A_{mn}\big(f^{k-n}(x)\big)
\Big)
\oslash_m\delta_m\Big]
\nonumber\\
&=&
g^{l-m}
\left(
\frac{{\rm D}_m A_{mn}\left(f^{k-n}(x)\right)}{{\rm D}_n f^{k-n}(x)}.
\nonumber
\right)
\ee

\subsection[\appendixname~\thesubsection]{Proof of (\ref{87''}):}

Define
\be
B_{lk}(x)
&=&
\int_a^x
A_{lk}(y) {\rm D}_k y
=
g^l
\left(
\int_{f^k(a)}^{f^k(x)}
A_{00}(r) {\rm d}r
\right)
=
g^l
\left(
B_{00}\big(f^k(x)\big)
\right)
,\\
C_{lk}(x)
&=&
g^{l-m}
\left(
\int_{f^{k-n}(a)}^{f^{k-n}(x)}
A_{mn}(y) {\rm D}_ny
\right)
=
g^{l-m}
\left(
C_{mn}\big(f^{k-n}(x)\big)
\right)
\ee
Now compute the derivatives:
\be
\frac{{\rm D}_l }{{\rm D}_k x}
B_{lk}(x)
&=&
\frac{{\rm D}_l }{{\rm D}_k x}
\int_a^x
A_{lk}(y) {\rm D}_k y
\\
&=&
\frac{{\rm D}_l }{{\rm D}_k x}
g^l
\left(
B_{00}\big(f^k(x)\big)
\right)
\\
&=&
g^l
\left(
\frac{{\rm d} }{{\rm d} f^k(x)}
B_{00}\big(f^k(x)\big)
\right)
\\
&=&
g^l
\left(
\frac{{\rm d} }{{\rm d} f^k(x)}
\int_{f^k(a)}^{f^k(x)}
A_{00}(r) {\rm d}r
\right)
\\
&=&
g^l
\left(
A_{00}\big(f^k(x)\big)
\right)=A_{lk}(x)
\ee

\vspace{-15pt}
\be
\frac{{\rm D}_l C_{lk}(x)}{{\rm D}_k x}
&=&
g^{l-m}
\left(
\frac{{\rm D}_m C_{mn}\left(f^{k-n}(x)\right)}{{\rm D}_n f^{k-n}(x)}
\right)
\\
&=&
g^{l-m}
\left(
\frac{{\rm D}_m }{{\rm D}_n f^{k-n}(x)}
\int_{f^{k-n}(a)}^{f^{k-n}(x)}
A_{mn}(y) {\rm D}_ny
\right)
\\
&=&
g^{l-m}
\left(
A_{mn}\big(f^{k-n}(x)\big)
\right)
=A_{lk}(x)
\ee
The derivatives are identical,
\be
\frac{{\rm D}_l }{{\rm D}_k x}
g^{l-m}
\left(
\int_{f^{k-n}(a)}^{f^{k-n}(x)}
A_{mn}(y) {\rm D}_ny
\right)
&=&
\frac{{\rm D}_l }{{\rm D}_k x}
\int_a^x
A_{lk}(y) {\rm D}_k y,
\ee
which implies
\be
g^{l-m}
\left(
\int_{f^{k-n}(a)}^{f^{k-n}(x)}
A_{mn}(y) {\rm D}_ny
\right)
=
\int_a^x
A_{lk}(y) {\rm D}_k y
\oplus_l \textrm{constant}
\ee
Setting $x=a$ we find
\be
0_l=0_l\oplus_l \textrm{constant}=\textrm{constant}
\ee
hence
\be
g^{l-m}
\left(
\int_{f^{k-n}(a)}^{f^{k-n}(x)}
A_{mn}(y) {\rm D}_ny
\right)
=
\int_a^x
A_{lk}(y) {\rm D}_k y
\ee
\subsection[\appendixname~\thesubsection]{Powers (Repeated Multiplications)}

\label{Monomials}

In order to introduce generalized arithmetics of complex numbers we need a useful concept of a ``first power'' \cite{BC}. To~this end, 
consider two sets $X$ and $Y$ and a map $A:X\to Y$ which can be described by a convergent power series. If~$X$ and $Y$ are equipped with different arithmetics we first have to clarify the meaning of ``power''. This will be performed as follows.
Consider two bijections $f_X:X\to {\mathbb R}$, $f_Y:Y\to {\mathbb R}$, and~their composition $f=f_Y^{-1}\circ f_X$. The~map 
\be
X\ni x\mapsto x^{1_{XY}} =f(x)\in Y
\ee
defines a first $Y$-valued power of $x\in X$. 

\begin{lemma}
\be
(x^{1_{XY}})^{1_{YZ}}
&=&
x^{1_{XZ}},\\
x^{1_{XY}1_{YX}}
&=&
x=x^{1_{XX}},\\
(x\odot_{X}y)^{1_{XY}}
&=&
x^{1_{XY}}\odot_{Y}y^{1_{XY}},\\
(x\oplus_{X}y)^{1_{XY}}
&=&
x^{1_{XY}}\oplus_{Y}y^{1_{XY}}.
\ee
\end{lemma}

\begin{proof}
\be
(x^{1_{XY}})^{1_{YZ}}
&=&
x^{1_{XY}1_{YZ}}
=
f_{Z}^{-1}\circ f_{Y}\circ f_{Y}^{-1}\circ f_{X}(x)
\nonumber\\
&=&f_{Z}^{-1}\circ f_{X}(x)=x^{1_{XZ}},\nonumber\\
(x\odot_{X}y)^{1_{XY}}
&=&
f_{Y}^{-1}\circ f_{X}(x\odot_{X}y)
=
f_{Y}^{-1}\big(f_{X}(x)f_{X}(y)\big)
\nonumber\\
&=&
f_{Y}^{-1}\Big(f_{Y}\big(x^{1_{XY}}\big)f_{Y}\big(y^{1_{XY}}\big)\Big)\nonumber\\
&=&
x^{1_{XY}}\odot_{Y}y^{1_{XY}},\nonumber\\
(x\oplus_{X}y)^{1_{XY}}
&=&
f_{Y}^{-1}\circ f_{X}(x\oplus_{X}y)
=
f_{Y}^{-1}\big(f_{X}(x)+f_{X}(y)\big)
\nonumber\\
&=&
f_{Y}^{-1}\Big(f_{Y}\big(x^{1_{XY}}\big)+f_{Y}\big(y^{1_{XY}}\big)\Big)
=
x^{1_{XY}}\oplus_{Y}y^{1_{XY}}.\nonumber
\ee
The remaining properties are obvious. 
\end{proof}

Let $n_{Y}=f_{Y}^{-1}(n)$, with~$n\in{\mathbb N}$ being a natural number satisfying the arithmetic defined by $f_{{\mathbb R}}(x)=x$. Then, first of all, 
\be
n_{Y} &=& f_{Y}^{-1}\circ f_{{\mathbb R}}(n)=n^{1_{{\mathbb R}Y}},\nonumber\\
n &=& n_{\mathbb R}.\nonumber
\ee
More generally,
\be
(n_{Y})^{1_{YZ}}
&=&
f_{Z}^{-1}\circ f_{Y}\circ f_{Y}^{-1}(n)=f_{Z}^{-1}(n)=n_{Z},\nonumber
\ee
and, in particular,
\be
(0_{Y})^{1_{YZ}}
&=&
f_{Z}^{-1}(0)=0_{Z},\nonumber\\
(1_{Y})^{1_{YZ}}
&=&
f_{Z}^{-1}(1)=1_{Z},\nonumber
\ee
are the relations between neutral elements in ${Y}$ and  ${Z}$.

An $n$th $Y$-valued power reads
\be
x^{n_{XY}} &=& x^{1_{XY}}\odot_{Y}\dots 
\odot_{Y} x^{1_{XY}} \quad \textrm{($n$ times)},\nonumber\\
&=&
f_{{Y}}^{-1}\Big(f_{{Y}}\big(x^{1_{XY}}\big)\dots f_{{Y}}\big(x^{1_{XY}}\big)\Big)\nonumber\\
&=&
f_{{Y}}^{-1}\Big(f_{{X}}(x)\dots f_{{X}}(x)\Big)\nonumber\\
&=&
f_{{Y}}^{-1}\Big(f_{{X}}\big(x\odot_{X}\dots\odot_{X} x\big)\Big)\nonumber\\
&=&
\big(x\odot_{X}\dots\odot_{X} x\big)^{1_{XY}},\nonumber\\
x^{n_{XY}} \odot_{Y} x^{m_{XY}}
&=&
x^{(n+m)_{XY}}.\nonumber
\ee
Note that one naturally arrives at the definition of
\be
x^{n_{X}}
&=&
x\odot_{X}\dots \odot_{X} x
\nonumber\\
&=&
 x^{1_{XX}}\odot_{X}\dots 
\odot_{X} x^{1_{XX}}
=
x^{n_{XX}},\nonumber
\ee
which coincides with the definition of $x^{n_k}$ discussed earlier. So, $n_{X}$, understood as a power, can be identified with $n_{X}=f_{X}^{-1}(n)$, since
\be
n_{X}\oplus_{X}m_{X}
&=&
f_{X}^{-1}(n+m)=(n+m)_{X},\nonumber
\ee
and thus one obtains the expected relation between products and sums,
\be
x^{n_{X}}\odot_{X}x^{m_{X}}=x^{(n+m)_{X}}=x^{n_{X}\oplus_{X}m_{X}}.\nonumber
\ee
Finally, let us compute
\be
x^{n_{XY}m_{YZ}}
&=&
\big(
\underbrace{
x\odot_{X}\dots\odot_{X} x}_n\big)^{1_{XY}1_{YZ}}
\odot_{Z}
\dots
\odot_{Z}
\big(
\underbrace{x\odot_{X}\dots\odot_{X} x}_n
\big)^{1_{XY}1_{YZ}}\nonumber\\
&=&
\big(x\odot_{X}\dots\odot_{X} x\big)^{1_{XZ}}\odot_{Z}
\dots
\odot_{Z}\big(x\odot_{X}\dots\odot_{X} x\big)^{1_{XZ}}\nonumber\\
&=&
\big(x^{1_{XZ}}\odot_{Z}\dots\odot_{Z} x^{1_{XZ}}\big)\odot_{Z}
\dots
\odot_{Z}\big(x^{1_{XZ}}\odot_{Z}\dots\odot_{Z} x^{1_{XZ}}\big)\nonumber\\
&=&
x^{1_{XZ}}\odot_{Z}\dots\odot_{Z} x^{1_{XZ}}
\quad \textrm{($nm$ times)}\nonumber\\
&=&
x^{(nm)_{XZ}}\nonumber
\ee
\subsection[\appendixname~\thesubsection]{Generalized Complex~Numbers}

\label{take two}

Let $(x,y)\in{X}_1\times {X}_2$. 
The arithmetic of complex numbers is defined as
\be
x\oplus y
&=&
(x_1,x_2)\oplus (y_1,y_2)
=
(x_1\oplus_{{X}_1}y_1,x_2\oplus_{{X}_2}y_2),\nonumber\\
x\odot y
&=&
(x_1,x_2)\odot (y_1,y_2)
\nonumber\\
&=&
\left(x_1\odot_{{X}_1}y_1\ominus_{{X}_1}x_2^{1_{{X}_2{X}_1}}\odot_{{X}_1}y_2^{1_{{X}_2{X}_1}},
x_1^{1_{{X}_1{X}_2}}\odot_{{X}_2}y_2 \oplus_{{X}_2} x_2\odot_{{X}_2}y_1^{1_{{X}_1{X}_2}} \right)
\nonumber
\ee
(we feel free to represent pairs of numbers as either columns or rows). Neutral elements of multiplication and addition are given by
\be
1_{X} &=& \left(1_{{X}_1},0_{{X}_2}\right),\\
0_{X} &=& \left(0_{{X}_1},0_{{X}_2}\right).
\ee
The ``imaginary unit'' is represented by
\be
i_{X} &=& \left(0_{{X}_1},1_{{X}_2}\right).
\ee
In order to simplify notation $i_{{X}}$ will be sometimes denoted by $i'$. 
We get the standard ``$i$ squared equals minus one'' rule,
\be
i'\odot (y_1,y_2)
&=&
\left(0_{{X}_1}\odot_{{X}_1}y_1\ominus_{{X}_1}1_{{X}_2}^{1_{{X}_2{X}_1}}\odot_{{X}_1}y_2^{1_{{X}_2{X}_1}},
0_{{X}_1}^{1_{{X}_1{X}_2}}\odot_{{X}_2}y_2 \oplus_{{X}_2} 1_{{X}_2}\odot_{{X}_2}y_1^{1_{{X}_1{X}_2}} \right)\nonumber\\
&=&
\left(\ominus_{{X}_1}y_2^{1_{{X}_2{X}_1}}, y_1^{1_{{X}_1{X}_2}} \right)=(z_1,z_2),\nonumber\\
i'\odot (z_1,z_2)
&=&
\left(\ominus_{{X}_1}z_2^{1_{{X}_2{X}_1}}, z_1^{1_{{X}_1{X}_2}} \right)
=
\left(\ominus_{{X}_1}y_1^{1_{{X}_1{X}_2}1_{{X}_2{X}_1}}, \left(\ominus_{{X}_1}y_2^{1_{{X}_2{X}_1}}\right)^{1_{{X}_1{X}_2}} \right)\nonumber\\
&=&
\left(\ominus_{{X}_1}y_1,\ominus_{{X}_2}y_2\right)=i'\odot i'\odot (y_1,y_2)=\ominus (y_1,y_2).\nonumber
\ee
Thinking of the plane as a representation of complex numbers, we identify real and imaginary parts as follows:
\be
\Re x
&=&
\left(
\begin{array}{c}
x_1
\\
0_{{X}_2}
\end{array}
\right),\nonumber\\
\Im x
&=&
\left(
\begin{array}{c}
x_2^{1_{{X}_2{X}_1}}
\nonumber\\
0_{{X}_2}
\end{array}
\right),\nonumber\\
i'\Im x
&=&
i'\odot
\left(
\begin{array}{c}
x_2^{1_{{X}_2{X}_1}}
\\
0_{{X}_2}
\end{array}
\right)
=
\left(
\begin{array}{c}
0_{{X}_1}
\\
x_2
\end{array}
\right),\nonumber
\ee
(the imaginary part is also real!). Decomposition of a general complex $x$ into its real and imaginary parts can be expressed in the usual way by means of addition,
\be
x &=&
\Re x\oplus i'\Im x
=
\left(
\begin{array}{c}
x_1\\
0_{{X}_2}
\end{array}
\right)
\oplus 
i'
\odot
\left(
\begin{array}{c}
x_2{}^{1_{{X}_2{X}_1}}\\
0_{{X}_2}
\end{array}
\right)
=
\left(
\begin{array}{c}
x_1\\
0_{{X}_2}
\end{array}
\right)
\oplus 
\left(
\begin{array}{c}
\ominus_{{X}_1}0_{{X}_2}{}^{1_{{X}_2{X}_1}}
\\
x_2{}^{1_{{X}_2{X}_1}1_{{X}_1{X}_2}}
\end{array}
\right)\nonumber\\
&=&
\left(
\begin{array}{c}
x_1\ominus_{{X}_1}0_{{X}_1}\\
0_{{X}_2}\oplus_{{X}_2} x_2
\end{array}
\right)
=
\left(
\begin{array}{c}
x_1\\
x_2
\end{array}
\right).\nonumber
\ee
Complex conjugation reads
\be
x^* &=&
\Re x\ominus i'\Im x
=
\left(
\begin{array}{c}
x_1\\
0_{{X}_2}
\end{array}
\right)
\ominus
i'
\odot
\left(
\begin{array}{c}
x_2{}^{1_{{X}_2{X}_1}}\\
0_{{X}_2}
\end{array}
\right)
=
\left(
\begin{array}{c}
x_1\\
0_{{X}_2}
\end{array}
\right)
\ominus
\left(
\begin{array}{c}
\ominus_{{X}_1}0_{{X}_2}{}^{1_{{X}_2{X}_1}}
\\
x_2{}^{1_{{X}_2{X}_1}1_{{X}_1{X}_2}}
\end{array}
\right)\nonumber\\
&=&
\left(
\begin{array}{c}
x_1\oplus_{{X}_1}0_{{X}_1}\\
0_{{X}_2}\ominus_{{X}_2} x_2
\end{array}
\right)
=
\left(
\begin{array}{c}
x_1\\
\ominus_{{X}_2} x_2
\end{array}
\right).\nonumber
\ee
Modulus squared is real,
\be
x\odot x^* 
&=&
\left(
\begin{array}{c}
x_1\\
x_2
\end{array}
\right)
\odot
\left(
\begin{array}{c}
x_1\\
\ominus_{{X}_2} x_2
\end{array}
\right)
=
\left(
\begin{array}{c}
x_1^{2_{{X}_1}}\oplus_{{X}_1} x_2^{2_{{X}_2{X}_1}}\\
\ominus_{{X}_2}x_1^{1_{{X}_1{X}_2}}\odot_{{X}_2}x_2 \oplus_{{X}_2} x_2\odot_{{X}_2}x_1^{1_{{X}_1{X}_2}}
\end{array}
\right)\nonumber\\
&=&
\left(
\begin{array}{c}
x_1^{2_{{X}_1}}\oplus_{{X}_1} x_2^{2_{{X}_2{X}_1}}\\
0_{{X}_2}
\end{array}
\right).\nonumber
\ee

The definition of addition is obvious, but~let us take a closer look at multiplication. Recalling that 
$x^{1_{{X}_2{X}_1}}=f^{-1}_{{X}_1}\circ f_{{X}_2}(x)$, 
$x^{1_{{X}_1{X}_2}}=f^{-1}_{{X}_2}\circ f_{{X}_1}(x)$,
we rephrase real and imaginary parts of the product as 
\be
(x\odot y)_1
&=&
f^{-1}_{{X}_1}
\Bigg(f_{{X}_1}(x_1)f_{{X}_1}(y_1)
-
f_{{X}_1}\left(x_2^{1_{{X}_2{X}_1}}\right)f_{{X}_1}\left(y_2^{1_{{X}_2{X}_1}}\right)
\Bigg)
\nonumber\\
&=&
f^{-1}_{{X}_1}
\Big(f_{{X}_1}(x_1)f_{{X}_1}(y_1)
-
f_{{X}_2}(x_2)f_{{X}_2}(y_2)
\Big)
=f^{-1}_{{X}_1}\big(\Re(\tilde x\tilde y)\big),\label{xy1}
\\
(x\odot y)_2
&=&
f^{-1}_{{X}_2}
\Bigg(
f_{{X}_2}\left(x_1^{1_{{X}_1{X}_2}}\right)f_{{X}_2}(y_2)
+
f_{{X}_2}(x_2)f_{{X}_2}\left(y_1^{1_{{X}_1{X}_2}} \right)
\Bigg)
\nonumber\\
&=&
f^{-1}_{{X}_2}
\Big(
f_{{X}_1}(x_1)f_{{X}_2}(y_2)
+
f_{{X}_2}(x_2)f_{{X}_1}(y_1)
\Big)
=
f^{-1}_{{X}_2}\big(\Im(\tilde x\tilde y)\big),
\label{xy2}
\ee
where $\tilde x=f_{{X}_1}(x_1)+if_{{X}_2}(x_2)=\tilde x_1+i\tilde x_2$, etc.
Analogously,
\be
(x\oplus y)_1
&=&
f^{-1}_{{X}_1}
\Big(f_{{X}_1}(x_1)+f_{{X}_1}(y_1)\Big)
=
f^{-1}_{{X}_1}\big(\Re(\tilde x+\tilde y)\big),
\nonumber\\
(x\oplus y)_2
&=&
f^{-1}_{{X}_2}
\Big(f_{{X}_2}(x_2)+f_{{X}_2}(y_2)\Big)
=f^{-1}_{{X}_2}\big(\Im(\tilde x+\tilde y)\big).
\nonumber
\ee
Accordingly,
\be
x\odot y
&=&
\Big(f^{-1}_{{X}_1}\big(\Re(\tilde x\tilde y)\big),f^{-1}_{{X}_2}\big(\Im(\tilde x\tilde y)\big)\Big),\\
x\oslash y
&=&
\Big(f^{-1}_{{X}_1}\big(\Re(\tilde x/\tilde y)\big),f^{-1}_{{X}_2}\big(\Im(\tilde x/\tilde y)\big)\Big),\\
x\oplus y
&=&
\Big(f^{-1}_{{X}_1}\big(\Re(\tilde x+\tilde y)\big),f^{-1}_{{X}_2}\big(\Im(\tilde x+\tilde y)\big)\Big),\\
x\ominus y
&=&
\Big(f^{-1}_{{X}_1}\big(\Re(\tilde x-\tilde y)\big),f^{-1}_{{X}_2}\big(\Im(\tilde x-\tilde y)\big)\Big).
\ee
One can still further simplify operations with complex numbers. Note that 
\be
\left(
\begin{array}{c}
x_1\\
x_2
\end{array}
\right)
&=&
\left(
\begin{array}{c}
f^{-1}_{{X}_1}\circ f_{{X}_1}(x_1)\\
f^{-1}_{{X}_2}\circ f_{{X}_2}(x_2)
\end{array}
\right)
=
\left(
\begin{array}{c}
f^{-1}_{{X}_1}(\Re \tilde x)\\
f^{-1}_{{X}_2}(\Im \tilde x)
\end{array}
\right),
\ee
so that 
\be
\left(
\begin{array}{c}
f^{-1}_{{X}_1}(\Re \tilde x)\\
f^{-1}_{{X}_2}(\Im \tilde x)
\end{array}
\right)
\odot
\left(
\begin{array}{c}
f^{-1}_{{X}_1}(\Re \tilde y)\\
f^{-1}_{{X}_2}(\Im \tilde y)
\end{array}
\right)
&=&
\left(
\begin{array}{c}
f^{-1}_{{X}_1}\big(\Re (\tilde x\tilde y)\big)\\
f^{-1}_{{X}_2}\big(\Im (\tilde x\tilde y)\big)
\end{array}
\right),\label{Simple 1}\\
\left(
\begin{array}{c}
f^{-1}_{{X}_1}(\Re \tilde x)\\
f^{-1}_{{X}_2}(\Im \tilde x)
\end{array}
\right)
\oplus
\left(
\begin{array}{c}
f^{-1}_{{X}_1}(\Re \tilde y)\\
f^{-1}_{{X}_2}(\Im \tilde y)
\end{array}
\right)
&=&
\left(
\begin{array}{c}
f^{-1}_{{X}_1}\big(\Re (\tilde x+\tilde y)\big)\\
f^{-1}_{{X}_2}\big(\Im (\tilde x+\tilde y)\big)
\end{array}
\right),\label{Simple 2}
\ee
which are, perhaps, the~most convenient forms of generalized complex arithmetic.
The operations $\Re$ and $\Im$ are denoted by identical symbols no matter which arithmetic is used. This will not lead to inconsistencies and should be clear from a context. The~standard Diophantine complex numbers are denoted either as $x=x_1+ix_2$ or $x=(x_1,x_2)$ and it is understood that the two notations mean the same, i.e.,~it is allowed to write \mbox{$x_1+ix_2=(x_1,x_2)$}.

\begin{lemma}
$\oplus$ and $\odot$ are associative and commutative, and~$\oplus$ is distributive with respect to $\odot$.
\label{complex nDA}
\end{lemma}
The proof is an immediate consequence of (\ref{Simple 1}) and (\ref{Simple 2}).

\subsection[\appendixname~\thesubsection]{Complex-Valued Scalar Product via Non-Newtonian Integration}

A complex-valued function is defined by the diagram
\be
\begin{array}{rcl}
{X}                & \stackrel{A}{\longrightarrow}       & {Y}={Y}_1\times {Y}_2               \\
f_{X}{\Big\downarrow}   &                                     & {\Big\downarrow}f_{Y}=f_{{Y}_1}\times   f_{{Y}_2} \\
{\mathbb R}                & \stackrel{\tilde A}{\longrightarrow}   & {\mathbb R}^2
\end{array}.
\label{l or r}
\ee
Consider two functions $A,B:{X}\to {{Y}_1}\times {{Y}_2}$ associated with diagrams of the form (\ref{l or r}). 
Let $\tilde A(r)=\tilde A_1(r)+i \tilde A_2(r)$, $\tilde B(r)=\tilde B_1(r)+i \tilde B_2(r)$, and~\be
\langle \tilde A|\tilde B\rangle
=
\int_{-f_{X}(T)/2}^{f_{X}(T)/2}
\overline{\tilde A(r)}\tilde B(r)
{\textrm d}r.\nonumber
\ee
$f_{X}(T)$ can be finite or infinite. 
The scalar product of two functions $A,B:{X}\to {{Y}_1}\times {{Y}_2}$ is defined as
\be
\langle A|B\rangle
&=&
\int_{\ominus_{{X}}T\oslash_{{X}} 2_{{X}}}
^{T\oslash_{{X}} 2_{{X}}}A(x)^*\odot B(x){\textrm D}x.
\ee
The same symbol of the scalar product for both $\langle \tilde A|\tilde B\rangle$ and $\langle A|B\rangle$ will not lead \mbox{to~ambiguities}.
\begin{lemma}
(Properties of the scalar product)
\be
\langle A|B\rangle
&=&
\left(
\begin{array}{c}
f^{-1}_{{Y}_1}\big(\Re\langle \tilde A|\tilde B\rangle\big)\\
f^{-1}_{{Y}_2}\big(\Im\langle \tilde A|\tilde B\rangle\big)
\end{array}
\right),\\
\langle A|B\rangle^*
&=&
\langle B|A\rangle,\\
\langle A|B\oplus C\rangle
&=&
\langle A|B\rangle\oplus \langle A|C\rangle,\\
\langle A|\lambda\odot B\rangle
&=&
\lambda\odot \langle A|B\rangle.
\ee
\end{lemma}

\begin{proof}

The integrand
\be
A(x)^*\odot B(x)
&=&
\left(
\begin{array}{c}
A_1(x)\\
\ominus_{{Y}_2}A_2(x)
\end{array}
\right)
\odot
\left(
\begin{array}{c}
B_1(x)\\
B_2(x)
\end{array}
\right)
\nonumber\\
&=&
\left(
\begin{array}{c}
f^{-1}_{{Y}_1}
\Big(f_{{Y}_1}\big(A_1(x)\big)f_{{Y}_1}\big(B_1(x)\big)
-
f_{{Y}_2}\big(\ominus_{{Y}_2}A_2(x)\big)f_{{Y}_2}\big(B_2(x)\big)
\Big)\\
f^{-1}_{{Y}_2}
\Big(
f_{{Y}_1}\big(A_1(x)\big)f_{{Y}_2}\big(B_2(x)\big)
+
f_{{Y}_2}\big(\ominus_{{Y}_2}A_2(x)\big)f_{{Y}_1}\big(B_1(x)\big)
\Big)
\end{array}
\right)
\nonumber\\
&=&
\left(
\begin{array}{c}
f^{-1}_{{Y}_1}
\Big(\tilde A_1\big(f_{{X}}(x)\big)\tilde B_1\big(f_{{X}}(x)\big)
+
\tilde A_2\big(f_{{X}}(x)\big)\tilde B_2\big(f_{{X}}(x)\big)
\Big)\\
f^{-1}_{{Y}_2}
\Big(
\tilde A_1\big(f_{{X}}(x)\big)\tilde B_2\big(f_{{X}}(x)\big)
-
\tilde A_2\big(f_{{X}}(x)\big)\tilde B_1\big(f_{{X}}(x)\big)
\Big)
\end{array}
\right)
\nonumber\\
&=&
\left(
\begin{array}{c}
f^{-1}_{{Y}_1}
\Big(\tilde C_1\big(f_{{X}}(x)\big)
\Big)\\
f^{-1}_{{Y}_2}
\Big(
\tilde C_2\big(f_{{X}}(x)\big)
\Big)
\end{array}
\right)
=C(x)
\nonumber.
\ee
{So,} %MDPI: Please center the equation label when equation is separated into more than one row.
{\small \be
\langle A|B\rangle
&=&
\int_{\ominus_{{X}}T\oslash_{{X}} 2_{{X}}}
^{T\oslash_{{X}} 2_{{X}}}C(x){\textrm D}x
=
\left(
\begin{array}{c}
f^{-1}_{{Y}_1}\left(\int_{f_{{X}}(\ominus_{{X}}T\oslash_{{X}} 2_{{X}})}^{f_{{X}}(T\oslash_{{X}} 2_{{X}})}\tilde C_1(r){\textrm d}r\right)\\
f^{-1}_{{Y}_1}\left(\int_{f_{{X}}(\ominus_{{X}}T\oslash_{{X}} 2_{{X}})}^{f_{{X}}(T\oslash_{{X}} 2_{{X}})}\tilde C_2(r){\textrm d}r\right)
\end{array}
\right)\nonumber\\
&=&
\left(
\begin{array}{c}
f^{-1}_{{Y}_1}\left(\int_{-f_{{X}}(T)/2}^{f_{{X}}(T)/2}\tilde C_1(r){\textrm d}r\right)\\
f^{-1}_{{Y}_2}\left(\int_{-f_{{X}}(T)/2}^{f_{{X}}(T)/2}\tilde C_2(r){\textrm d}r\right)
\end{array}
\right)
=
\left(
\begin{array}{c}
f^{-1}_{{Y}_1}\left(\int_{-f_{{X}}(T)/2}^{f_{{X}}(T)/2}\Re\left(\overline{\tilde A(r)}\tilde B(r)\right){\textrm d}r\right)\\
f^{-1}_{{Y}_2}\left(\int_{-f_{{X}}(T)/2}^{f_{{X}}(T)/2}\Im\left(\overline{\tilde A(r)}\tilde B(r)\right){\textrm d}r\right)
\end{array}
\right)
\\
&=&
\left(
\begin{array}{c}
f^{-1}_{{Y}_1}\big(\Re\langle \tilde A|\tilde B\rangle\big)\\
f^{-1}_{{Y}_2}\big(\Im\langle \tilde A|\tilde B\rangle\big)
\end{array}
\right).\nonumber
\ee}
By definition of complex conjugation
{\small\be
\langle A|B\rangle^*
&=&
\left(
\begin{array}{c}
f^{-1}_{{Y}_1}\big(\Re\langle \tilde A|\tilde B\rangle\big)\\
\ominus_{{Y}_2}f^{-1}_{{Y}_2}\big(\Im\langle \tilde A|\tilde B\rangle\big)
\end{array}
\right)
=
\left(
\begin{array}{c}
f^{-1}_{{Y}_1}\big(\Re\langle \tilde A|\tilde B\rangle\big)\\
f^{-1}_{{Y}_2}\big(-\Im\langle \tilde A|\tilde B\rangle\big)
\end{array}
\right)
=
\left(
\begin{array}{c}
f^{-1}_{{Y}_1}\big(\Re\langle \tilde B|\tilde A\rangle\big)\\
f^{-1}_{{Y}_2}\big(\Im\langle \tilde B|\tilde A\rangle\big)
\end{array}
\right)\nonumber\\
&=&
\langle B|A\rangle.
\ee}The remaining two properties follow from associativity and distributivity of the arithmetic operations, supplemented by linearity of { the integral.} %MDPI: Please check whether here is the end of proof.
\end{proof}
\subsection[\appendixname~\thesubsection]{Continuous Transition Between Two Levels of the~Hierarchy}

There exists a simple generalization of the hierarchies,  allowing for a continuous transition between the levels. It is based on the~following
\begin{lemma}
Consider a collection (finite or not) of parameters $\lambda_j$, $\sum_j\lambda_j=1$, and~a collection of bijections $g_j$ that satisfy  Lemma~1. Then
$
g=\sum_j\lambda_j g_j
$
also satisfies  Lemma~1.
\end{lemma}
\begin{proof}
Each $g_j$ can be written as
\be
g_j(p)=\frac{1}{2} + h_j\left(p-\frac{1}{2}\right),
\ee
where $h_j(-x)=-h_j(x)$, so
\be
g(p)=
\sum_j\lambda_j g_j(p)
=
\underbrace{\frac{1}{2}\sum_j\lambda_j}_{1/2}  + \underbrace{\sum_j\lambda_j h_j\left(p-\frac{1}{2}\right)}_{h\left(p-1/2\right)}.
\ee
The function 
$
h(x)=\sum_j\lambda_j h_j(x)
$
is odd, as~a linear combination of odd functions.
\end{proof}

As a side remark, let us note that  $g(p)=\sin^2\frac{\pi}{2}p$ satisfies all the concavity and monotonicity properties required for the existence of non-integer iterations $g^r$, $r\not\in\mathbb{Z}$, (cf.~\cite{Kuczma}, Chapter XV). The~problem is intriguing as an alternative possibility of continuously switching between quantum and subquantum levels of the hierarchy (cf. the discussion of zero-multiplier iterative roots  in~\cite{KCG}, Chapter 11.5).

%BBBBBBBBBBBBBBBBBBBBBBBBBBBBBBBBBBBBBBBBBBBBBBBBBBBBBBBB

\end{document}